\documentclass{article}[12pt]
\makeatletter
\renewcommand\paragraph{\@startsection{paragraph}{4}{\z@}%
                                      {1ex \@plus1ex \@minus.2ex}%
                                      {-1em}%
                                      {\normalfont\normalsize\bfseries}}
\makeatother

\usepackage{graphicx} 
\usepackage{CJK}
\usepackage{authblk}

\usepackage[margin=1in]{geometry}

\usepackage{mathpazo}

\usepackage{enumerate}

\usepackage{microtype}

\usepackage{tikz}
\usetikzlibrary{backgrounds,fit,decorations.pathreplacing,calc,quantikz}

\usepackage{soul}

\usepackage{float}

\usepackage[colorlinks = true]{hyperref}
\definecolor{darkred}  {rgb}{0.5,0,0}
\definecolor{darkblue} {rgb}{0,0,0.5}
\definecolor{darkgreen}{rgb}{0,0.5,0}
\hypersetup{
  urlcolor   = blue,         
  linkcolor  = darkblue,     
  citecolor  = darkgreen,    
  filecolor  = darkred       
}

\usepackage{amsmath,amssymb,amsfonts,amsthm,amstext}
\usepackage{bbm}

\usepackage{etoolbox}

\usepackage{mathtools}
\mathtoolsset{centercolon}
\makeatletter
\protected\def\tikz@nonactivecolon{\ifmmode\mathrel{\mathop\ordinarycolon}\else:\fi}
\makeatother

\usepackage{tcolorbox}

\usepackage{algorithm}
\usepackage{algpseudocode}

\usepackage{cleveref}
\usepackage{caption}
\usepackage{subcaption}

\crefname{lemma}{Lemma}{Lemmas}
\crefname{proposition}{Proposition}{Propositions}
\crefname{definition}{Definition}{Definitions}
\crefname{theorem}{Theorem}{Theorems}
\crefname{conjecture}{Conjecture}{Conjectures}
\crefname{corollary}{Corollary}{Corollaries}
\crefname{claim}{Claim}{Claims}
\crefname{section}{Section}{Sections}
\crefname{appendix}{Appendix}{Appendices}
\crefname{figure}{Fig.}{Figs.}
\crefname{table}{Table}{Tables}

\usepackage[retainorgcmds]{IEEEtrantools}

\usepackage{tocloft}

\usepackage{multirow}

\usepackage{mathtools}

\usepackage{tikz}	
\usetikzlibrary{backgrounds,fit,decorations.pathreplacing}

\usepackage{algorithm}
\usepackage{algpseudocode}

\newcommand{\x}{\otimes}

\DeclareMathOperator{\Tab}{Tab}
\DeclareMathOperator{\Obl}{Obl} 

\DeclarePairedDelimiter{\set}{\lbrace}{\rbrace}
\DeclarePairedDelimiter{\abs}{\lvert}{\rvert}
\DeclarePairedDelimiter{\norm}{\lVert}{\rVert}

\newcommand{\poly}{\mathrm{poly}}
\newcommand{\polylog}{\mathrm{polylog}}
\newcommand{\polyloglog}{\mathrm{polyloglog}}
\newcommand{\polylogloglog}{\mathrm{polylogloglog}}

\newcommand{\C}{\mathbb{C}}

\newcommand{\Z}{\mathbb{Z}}
\newcommand{\calH}{\mathcal{H}}
\newcommand{\calX}{\mathcal{X}}
\newcommand{\calY}{\mathcal{Y}}
\newcommand{\calA}{\mathcal{A}}
\newcommand{\calB}{\mathcal{B}}

\newcommand{\1}{\mathbb{1}}

\newcommand{\phip}{\phi_+}
\newcommand{\id}{\mathbb{1}}
\newcommand{\EPR}{\mathrm{EPR}}

\newcommand{\pr}[2]{P(#1|#2)}

\newcommand{\PCPP}{\mathsf{PCPP}}

\newcommand{\ipt}{\mathsf{input}}

\newcommand{\MIP}{\mathsf{MIP}}
\newcommand{\pzkMIP}{\mathsf{PZK}\text{-}\MIP}
\newcommand{\szkMIP}{\mathsf{SZK}\text{-}\MIP}
\newcommand{\czkMIP}{\mathsf{CZK}\text{-}\MIP}
\newcommand{\RE}{\mathsf{RE}}
\newcommand{\val}{\mathrm{val}}
\newcommand{\NEXP}{\mathsf{NEXP}}

\newcommand{\NTIME}{\mathsf{NTIME}}
\newcommand{\Sim}{\mathsf{Sim}}
\newcommand{\Graph}{\mathcal{G}}
\newcommand{\QMIP}{\mathsf{QMIP}}
\newcommand{\IP}{\mathsf{IP}}
\newcommand{\OR}{\mathrm{OR}}
\newcommand{\QR}{\mathrm{QR}}
\newcommand{\AR}{\mathrm{AR}}
\newcommand{\PR}{\mathrm{PR}}

\newcommand{\appd}[1]{\simeq_{#1}}

\newcommand{\hft}[1]{\textcolor{darkgreen}{#1}}

\newtheorem{theorem}{Theorem}[section]
\newtheorem{lemma}[theorem]{Lemma}
\newtheorem{proposition}[theorem]{Proposition}
\newtheorem{definition}[theorem]{Definition}

\newtheorem*{conjecture*}{Conjecture}

\theoremstyle{definition}

\begin{document}

\title{Succinct Perfect Zero-knowledge for MIP*}
\date{}
\author[1]{Honghao Fu~\thanks{Email: \href{mailto:honghao.fu@concordia.ca}{honghao.fu@concordia.ca}}}
\author[2]{Kieran Mastel~\thanks{Email: \href{mailto:kmastel@uottawa.ca}{kmastel@uottawa.ca}}}
\author[3]{Xingjian Zhang~\thanks{Email: \href{mailto:zxj24@hku.hk}{zxj24@hku.hk}}}

\renewcommand*{\Authfont}{\normalsize}
\affil[1]{ Concordia Institute for Information Systems Engineering \\ Concordia University, Montreal, Canada}
\affil[2]{Department of Mathematics and Statistics, University of Ottawa, Ottawa, Canada}
\affil[3]{ QICI Quantum Information and Computation Initiative\\
School of Computing and Data Science, The University of Hong Kong, Hong Kong SAR, China}

\maketitle

\begin{abstract}
     In the recent breakthrough result \cite{mastel2024two}, Slofstra and the second author show that there is a two-player one-round perfect zero-knowledge $\MIP^*$ protocol for $\RE$. We build on their result to show that there exists a succinct two-player one-round perfect zero-knowledge $\MIP^*$ protocol for $\RE$ against dishonest verifiers with $\polylog$ question size and $O(1)$ answer size, or with $O(1)$ question size and $\polylog$ answer size. To prove our result, we study the three central compression techniques underlying the $\MIP^*=\RE$ proof \cite{re}---question reduction, oracularization, and answer reduction. We show that question reduction preserves the perfect (as well as statistical and computational) zero-knowledge properties of the original protocol against dishonest verifiers, and oracularization and answer reduction preserve the perfect (as well as statistical and computational) zero-knowledge properties of the original protocol against honest verifiers. Secondly, we show that every constraint-constraint binary constraint system (BCS) nonlocal game, which provides a quantum information characterization of $\MIP^*$, can be converted to a synchronous constraint-variable BCS game to preserve perfect completeness for our compression. Lastly, we present a parametrized perfect-zero-knowledge transformation of $\MIP^*$ protocols, which generalizes the transformation in \cite{mastel2024two}. This transformation allows us to preserve the zero-knowledge property against dishonest verifiers in the recursively oracularized protocols in our compression. 
\end{abstract}
\tableofcontents

\section{Introduction}
An interactive proof system involves two parties: a prover with unlimited resources for computation and a verifier limited to bounded computation power. The prover aims to prove a statement to the verifier through rounds of message exchanges. Nevertheless, the prover might be dishonest and prove false statements to the verifier~\cite{goldwasser1985knowledge}. 
An important discovery in complexity theory is that the class of interactive proof systems, $\mathsf{IP}$, is equal to $\mathsf{PSPACE}$, where a polynomial time verifier can verify any statement that can be computed in polynomial space, i.e., in $\mathsf{PSPACE}$ \cite{shamir1992ip}.
Moreover, the interactive proof system can be turned into a zero-knowledge proof system ~\cite{goldwasser1985knowledge}. That is, the prover ascertains the statement to the verifier without unveiling any information of the proof beyond this mere fact. 
More precisely, it means that there is a simulator that does not have the proof but can still simulate the distribution of an honest prover's answers.
Assuming the existence of one-way functions, it is known that $\mathsf{IP}=\mathsf{CZK}$ \cite{goldwasser1985knowledge}. The complexity class $\mathsf{CZK}$ contains all the problems admitting \emph{computational zero-knowledge} proofs, where the simulated distribution is computationally indistinguishable from an honest one.
However, it is not yet clear whether we can get rid of the cryptographic assumptions to establish \emph{statistical zero-knowledge} ($\mathsf{SZK}$),
where the simulated distribution is statistically close to an honest one,
or \emph{perfect zero-knowledge} ($\mathsf{PZK}$), where the simulated distribution equals an honest one.

A remarkable solution to remove the computationally intractable assumptions is to consider a setting with multiple non-communicating provers~\cite{ben1988multi}. This is the setting of the complexity class $\mathsf{MIP}$, \emph{multi-prover interactive proof} systems.
In this setting, every interactive proof between the verifier and the provers can be turned into a perfect zero-knowledge one, i.e., $\mathsf{MIP}=\pzkMIP$ \cite{ben1988multi}. 
Since this setting does not involve any computationally intractable assumption, its efficiency is significantly improved over the ones built on such assumptions.
Moreover, as the verifier can communicate with multiple provers separately and cross-check their proofs, the verifier's capability is also strengthened. 
As shown by the celebrated results of Babai, Fortnow and Lund, we have $\mathsf{MIP}=\mathsf{NEXP} \supseteq \mathsf{PSPACE}$~\cite{babai1991non}. 
The $\MIP$ protocol for $\NEXP$ can be quite simple, with only two provers and one round of interaction. 

In $\mathsf{MIP}$, the resource between the provers is shared randomness.
What if the provers can share quantum resources such as entanglement? Moreover, what if the verifier can also run polynomial time quantum computation and sends quantum messages?
The former question led to the definition of $\MIP^*$ \cite{cleve2004consequences}, where the provers can share unlimited entanglement while the verifier is still classical.
The latter question inspired the definition of $\QMIP$ \cite{qmip}, where the verifier is also quantum.
Although a classical verifier might seem weaker than the quantum one, they can utilize \emph{self-tests} to verify whether the provers are behaving honestly \cite{mayers1998quantum}, which are tools absent in the classical case. 
This observation led to the breakthrough result $\QMIP = \MIP^*$ \cite{ruv13}.
However, the full computational power of $\MIP^*$ was still elusive.
It turns out that the law of quantum mechanics allows the classical verifier to command the quantum provers with very low communication \cite{neexp}, which led to
the celebrated result of $\mathsf{MIP}^*=\mathsf{RE}$~\cite{re}. This result says that any language in $\RE$ including the undecidable Halting problem has an $\MIP^*$ protocol.
Similar to its classical counterpart for any $\NEXP$ language, an $\MIP^*$ protocol with only two provers and one round of interaction suffices for any $\RE$ language \cite{re}.

To get zero-knowledge $\MIP^*$ protocols, Chiesa and his colleagues first proved that every language in $\NEXP$ has a perfect zero-knowledge $\MIP^*$ protocol with two provers but polynomially many rounds of interactions,
so $\NEXP \subseteq \pzkMIP^*$ \cite{chiesa2022spatial}.
On the other hand, Grilo and his colleagues proved that every language in $\MIP^*$ also has a six-prover one-round $\pzkMIP^*$ protocol, so $\MIP^* = \pzkMIP^*$~\cite{grilo2019perfect}.
The latest remarkable result of Mastel and Slofstra brings the number of provers from six down to two and proves that every language in $\RE$ has a two-prover one-round $\pzkMIP^*$ protocol \cite{mastel2024two}.

The line of works above has reduced the numbers of provers and rounds of interactions to the minimal, but the parameters such as the communication cost and the completeness-soundness gap are not fully optimized. 
In the work of \cite{mastel2024two}, the questions and answers must be polynomially long to ensure a constant completeness-soundness gap. Otherwise, if the question length is polylog and the answer length is a constant, the completeness-soundness gap is inverse-polynomial.
In a recent work by Culf and Mastel, they succeeded in reducing the answer length to a constant while still requiring polynomial-length questions \cite{culf2024re}.
Therefore, one of the open problems is how to construct \emph{succinct} $\pzkMIP^*$ protocols with polylog question length and constant answer length for $\RE$ with a constant completeness-soundness gap, or more generally, to construct \emph{succinct} $\pzkMIP^*$ protocols of which the total communication between the verifier and the provers is of a polylog length.

The motivation to explore succinct $\pzkMIP^*$ protocols resides in several aspects.
For the validity of multi-prover interactive proofs (both classical and quantum ones), a key issue is to prohibit the provers from communicating with each other. 
As first pointed out by Kilian, the physical implementation requirements can be reduced if the communication cost and the verification time can be reduced~\cite{kilian1990strong}.
Usually to shield potential information leakage, one may take advantage of the principles of relativity~\cite{bell2004speakable}:
The verifier is posed between the provers, and the provers are spatially separated far enough
so that they cannot exchange information before their answers reach the verifier due to the light speed restriction.
Therefore, as the communication cost gets lower and the verification time gets shorter, the spatial configuration can be made more compact and the spatial separation condition is enforced more easily.
Besides the motivation to validate a multi-prover setting, succinctness has always been desirable for practical applications of zero-knowledge proofs~\cite{kilian1992note}.
For instance, a famous usage of zero-knowledge proofs is to protect privacy in blockchains. As practical blockchain applications usually involve many users over the internet, a low communication cost in its implementation is vital for high-speed information transmission on a large scale \cite{goldwasser2015delegating}.

In classical theoretical computer science, building on techniques such as the probabilistic checkable proofs (PCPs) \cite{babai1991checking,feige1996interactive,arora1998probabilistic,arora1998proof}, the pursuit for succinctness in $\MIP$ protocols and zero-knowledge proofs led to the following result~\cite{dwork1992low}:
Every language in $\mathsf{NP}$ admits a perfect zero-knowledge proof with two classical provers and one round of interaction, which has perfect completeness, a constant complete-soundness gap, a logarithmic question length, and a constant answer length.
In this paper, we show that similar parameters can be obtained for $\pzkMIP^*$ protocols for $\RE$.
\begin{theorem}\label{thm:main}
    Let $\MIP^*[q,a]$ denote the two-prover one-round $\MIP^*$ protocols with length-$q$ questions and length-$a$ answers.
    Then every language in $\RE$ has a succinct perfect zero-knowledge two-prover one-round $\MIP^*[\polylog, O(1)]$ protocol with perfect completeness, constant soundness, and polynomial sampling and verification time.
\end{theorem}
Following a similar argument, we also get the following.
\begin{theorem}
    \label{thm:am}
    Every language in $\RE$ has a succinct perfect zero-knowledge two-prover one-round $\MIP^*[O(1),\polylog]$ protocol with perfect completeness, constant soundness, constant sampling time, and polynomial verification time, which is also an $\mathrm{AM}^*(2)$ protocol.
\end{theorem}

\paragraph{Proof overview.}

Recall that the proof of $\MIP^* = \RE$ is built upon four key techniques used in the compression of $\MIP^*$ protocols: question reduction, oracularization, answer reduction and parallel repetition.
These four techniques are recursively applied to families of $\MIP^*$ protocols to construct an $\MIP^*$ protocol for the $\RE$-complete Halting problem.
Question reduction is first applied to reduce the question length exponentially, but it will reduce the completeness-soundness gap from a constant to inverse polylog. 
The central idea behind question reduction is letting the provers sample their own questions and prove to the verifier that their questions are sampled from the correct distribution. Since the verifier does not need to send the original questions, the question length is reduced.
Then, oracularization is applied to ensure the protocol is suitable for answer reduction: The oracularized prover answers the questions to the two provers. This step preserves the completeness-soundness gap.
The third step is to apply answer reduction, where the oracularized prover computes a succinct proof, i.e., a probabilistic checkable proof (PCP), that his/her answers satisfy the verifier's checks. 
Instead of reading the entire answers, the verifier only needs to query bits of the answers and the succinct proof, which also reduces the answer length and verification time exponentially.
In the end, parallel repetition is applied to restore the constant completeness-soundness gap.

In this work, we build on the $\pzkMIP^*$ protocol with a constant completeness-soundness gap and polynomial-length questions and answers in \cite{mastel2024two} and construct a succinct $\pzkMIP^*$ protocol. For this purpose, we first examine if the four techniques towards proving $\MIP^*=\mathsf{RE}$ preserve the zero-knowledge property. Previously, \cite{mastel2024two} has shown that parallel repetition preserves the zero-knowledge property of the original protocol. We further show that question reduction in \cite{re} preserve the zero-knowledge property.
The subtlety is that we can only show that oracularization and 
 the tightened answer reduction technique of \cite{dong2023computational} preserve zero-knowledge against \emph{honest verifiers}, who only sample questions that are samplable in the original protocol.
 The subtlety will be discussed in more details later.
As a remark, 
the tightened answer reduction technique of \cite{dong2023computational} preserves the completeness-soundness gap and can reduce answer size exponentially.
Moreover, it is capable of reducing answer size below $\polylog(n)$, whereas the answer reductions in \cite{re,natarajan2023quantum} cannot.

The idea behind proving that question reduction preserves zero-knowledge is the observation that the honest distribution of the question-reduced protocol is the product of two distributions: the first one is obtained from measuring many copies of $\ket{\text{EPR}} = \frac{1}{\sqrt{2}}(\ket{00}+\ket{11})$ in the Pauli Z or X basis, and the second one is the honest distribution of the protocol before question reduction. Since simulating Pauli measurements on EPR pairs can be done in polynomial time \cite{aaronson2004improved}, if the honest distribution of the protocol before question reduction is simulable by a poly-time simulator, so is the new product distribution.
The idea behind proving that answer reduction preserves zero-knowledge is to notice that the honest prover will first sample answers from the honest distribution of the protocol before answer reduction, compute the PCP of the correctness of the answers, and then answer the new questions from the verifier based on the PCP.
Since in the complete case, the PCP can be computed in polynomial time, all that the provers need to do is efficient classical post-processing.
Hence, if a simulator can simulate the honest distribution of the original protocol, the new simulator for the answer-reduced protocol only needs to perform efficient classical post-processing on the output of the original simulator.

However, a subtle issue is the oracularization step. While we show this procedure preserves the zero-knowledge property against an honest verifier, in general, it does not preserve the perfect zero-knowledge property against a \emph{dishonest verifier}. 
Instead of specifying only one prover as the oracularized prover, a dishonest verifier may ask the two provers different oracle questions and learn more information in a round than an honest verifier. 
Even if a simulator exists for the base game, it is unclear how to construct a new simulator for the oracularized game
that can produce the correlation of an honest quantum strategy when asked two oracularized question.
Since answer reduction is built upon oracularization, we cannot prove answer reduction preserve zero-knowledge
against dishonest verifiers in general.

To mitigate this problem in the proof of \cref{thm:main}, we propose a parametrized perfect zero-knowledge transformation (\cref{lem:tab}) for the oracularized nonlocal game, which combines a technique called obliviation and Barrington's branching program \cite{Barrington86}. This transformation generalizes the perfect zero-knowledge protocol in \cite{mastel2024two}, which was inspired by \cite{dwork1992low} that constructs a succinct perfect zero-knowledge protocol for the NP language in the $\MIP$ setting. 
We add parameters for the number of rows in the tableau and the degree of obliviation. 
We choose the parameters to be large enough that the dishonest verifier cannot learn the assignment to a non-oblivious variable even by cheating and asking two oracle questions. 
In other words, the answers from the honest disttribution are completely random to the dishonest verifier, which makes the answers samplable in polynomial time.
We remark that after performing this transformation once on the baseline protocol, recursive oracularization shall preserve the perfect zero-knowledge property.
Moreover, when the two parameters are constant, the transformation preserves constant soundness.


To prove \cref{thm:main}, there is one more technical point before we apply the four techniques to the $\pzkMIP^*$ protocol for $\RE$ with a constant completeness-soundness gap in \cite{mastel2024two}.
In this $\pzkMIP^*$ protocol, the verifier randomly selects two constraints independently and distributes them to the two provers. We call such a protocol a constraint-constraint protocol. However, answer reduction cannot be applied to such a protocol, because for answer reduction to preserve perfect completeness, the strategy for the complete case should be a commuting strategy. That is, Alice and Bob's measurement operators should commute if the two corresponding questions have a nonzero probability of being sampled by the verifier. 
To fix this issue, we need to convert the protocol to a constraint-variable protocol, where the verifier first samples a random constraint and then samples a random variable from the constraint. 
\begin{theorem}[Informal]
    Let $P$ be an instance of a constraint-constraint protocol with $n$ variables and $m$ constraints in total,
    and let $P^{cv}$ be the corresponding instance of a constraint-variable protocol.
    \begin{itemize}
        \item If $P$ has a perfect quantum strategy, $P^{cv}$ has a perfect \emph{commuting} strategy.
        \item If the quantum winning probability of $P^{cv}$ is $1 - \epsilon$, the quantum winning probability of $P$ is at least $1 - \poly(\epsilon)$.
        \item If a perfect distribution of $P$ can be simulated, so can a perfect distribution of $P^{cv}$.
    \end{itemize}
\end{theorem}
This transformation is optimal because it preserves the question and answer size, perfect completeness, the zero-knowledge property, and more importantly, the constant soundness.
To prove the soundness of this transformation, we first use results from \cite{marrakchi2023almost} to show the shared state can be assumed to be the EPR pairs. 
By carefully design the new question distribution, we ensure this step preserve constant soundness.
Then we show the constraint measurements from a near-perfect strategy of $P^{cv}$ can be used by both parties in $P$ because these measurements must be consistent with the variable measurements no matter which side they are applied on.
The second step also preserves constant soundness.

To prove \Cref{thm:main}, we start with an $\MIP^*[O(1),\polylog]$ protocol from \cite{natarajan2023quantum}, which does not yet have the zero-knowledge property, and the nonlocal games are constraint-constraint games. We first apply the parametrized perfect zero-knowledge transformation to the protocol and transform the constraint-constraint protocol to a constraint-variable protocol. 
Since this step reduces the completeness-soundness gap, we apply a parallel repetition to restore it to a constant. Afterwards, we apply question reduction and recursive answer reduction until the protocol has a constant answer length. In the end, we get a succinct $\pzkMIP^*$ for $\RE$ with $\polylog$-question length, $O(1)$-answer length,
polynomial verification and sampling time, perfect completeness and constant soundness.


\paragraph{Discussion and open problems.}
In this paper, we show that every language in $\RE$ has a succinct $\pzkMIP^*$ protocol with constant completeness-soundness gap.
To prove this result, we prove that question reduction preserves the perfect (resp. statistical and computational) zero-knowledge property of the original protocol against a dishonest verifier. 
We also prove that oracularization and answer reduction preserve the perfect (resp. statistical and computational) zero-knowledge property of the original protocol against an honest verifier.
Additionally, we propose a parametrized perfect zero-knowledge transformation for the oracularized nonlocal game.
As the completeness and soundness properties of these techniques have been established in previous works, our result adds the preservation of zero-knowledge properties to their strengths, making these techniques more suitable for $\MIP^*$ and quantum cryptography studies.

Moreover, we show that every constraint-constraint binary constraint system (BCS) nonlocal game can be converted
to a synchronous constraint-variable BCS game while preserving the perfect completeness and soundness.
Note that BCS games are equivalent to $\MIP^*$ protocols \cite{mastel2024two}.
The other direction of converting constraint-variable protocols to constraint-constraint protocols was first proved in \cite{paddock2022rounding}.
Upon finishing this paper, we notice the conversion from constraint-constraint protocols to constraint-variable protocols was independently proved in \cite{culf2024re}. They proved the result using operator algebra techniques, and we take a quantum information approach in the proof.

We list a few open problems as future research directions.
\begin{enumerate}
    \item Can oracularization preserve the perfect zero knowledge property against a dishonest verifier? Notice that a dishonest verifier of an oracularized protocol can ask two independent oracularized questions to the two provers individually. 
    A technical challenge is that the corresponding observables may not be compatible, and it is hence not immediately clear how to efficiently simulate the correlated measurement statistics even given a simulator that can simulate the answers of one orcularized questions.
    We expect additional transformations to be added to oracularization to make it preserve zero knolwedge against 
    dishonest verifiers.
    In this work, we only show that this compression technique preserves zero knowledge against an honest verifier, and we need to develop the parameterized PZK transformation for the final target of succinct $\pzkMIP^*$ against a dishonest verifier. Nevertheless, it is worth mentioning that once this additional transformation is performed, one can recursively apply oracularization to the transformed protocol, which will then preserve the perfect zero-knowledge property. 
    
    \item How do we compress $\MIP^*$ protocols with imperfect completeness? The completeness guarantees of question reduction, oracularization, and answer reduction break down when the original protocol has imperfect completeness. 
    Classically, by giving up perfect completeness, H\aa stad designed a new PCP for $\mathsf{NP}$ with constantly many queries, of which the verifier only needs to perform binary addition \cite{haastad2001some}.
    If dropping the perfect completeness condition of $\MIP^*$ protocols can give us new $\MIP^*$ protocols for quantum computational complexity classes like $\mathsf{QMA}$ with improved parameters, we need new techniques to compress such protocols.
    \item Can we find $\pzkMIP^*$ (resp. $\szkMIP^*$, $\czkMIP^*$) protocols for smaller quantum computational complexity classes with improved parameters? The game version of the quantum PCP conjecture is about finding a succinct $\MIP^*$ protocol for a $\mathsf{QMA}$-complete problem \cite{gameqpcp}.
    For smaller complexity classes, we may utilize the structure therein to design specific zero-knowledge $\MIP^*$ protocols instead of using the generic protocol for $\RE$. In the classical literature, zero-knowledge protocols for $\mathsf{P}$ have better parameters than those of $\mathsf{NP}$~\cite{goldwasser2015delegating}.
    
    \item Can $\MIP^*$ protocols help transform interactive proof system into non-interactive ones? For single-prover interactive proof protocols, namely $\IP$, the Fiat-Shamir heuristic provides a paradigm to transform them into non-interactive proofs~\cite{fiat1986prove}.
    In practice, this line of research breeds succinct non-interactive arguments (SNARGs) and non-interactive succinct arguments of knowledge (SNARKs)~\cite{groth2010short,parno2016pinocchio,ben2013snarks,braun2013verifying,ben2014succinct}. Moreover, they can usually retain zero-knowledge properties. While the random oracle is usually adopted in the transformation, it has been realized that the study of $\MIP$ can help construct non-interactive protocols relying on more standard cryptographic assumptions~\cite{boneh2018quasi}. 
    Then, can we also do it in the quantum world?

    \item Can we further reduce the question length from polylog to log?
    The lower bound of the communication cost of $\MIP^*$ protocols for $\RE$ is log instead of polylog \cite{natarajan2023quantum}. To further reduce the question length to log, we possibly need new ideas for question reduction. Note that this problem of reducing question length to log is also related to the game version of the quantum PCP conjecture~\cite{gameqpcp}.
  
\end{enumerate}


\paragraph{Acknowledgments.} We would like thank Anand Natarajan and Yilei Chen for helpful discussions.
X.Z. acknowledges support from Hong Kong Research Grant Council through grant number R7035-21 of the Research Impact Fund and No.27300823. K.M. acknowledges support from the Natural Sciences and Engineering Research Council of Canada (NSERC).

\section{Preliminary}
\label{sec:prelim}

\paragraph{Notations.}
We denote the set $\set{1,\ldots, n}$ by $[n]$.
We denote a d-dimensional maximally entangled state (MES) by $\ket{\phi_+} = \frac{1}{\sqrt{d}}\sum_{i=1}^d \ket{i,i}$.
The special case of $d=2$ is the EPR pair, which is denoted by $\ket{\text{EPR}} = \frac{1}{\sqrt{2}}(\ket{00}+\ket{11})$.
We denote a Hilbert space by $\calH$.
For approximations, we use the following notations.
If $a,b \in \C$, we write $a \appd{\epsilon} b$ to mean $\abs{a-b} \leq \epsilon$.
For vectors $\ket{v}$ and $\ket{u}$, we write $\ket{u} \appd{\epsilon} \ket{v}$ to mean $\norm{\ket{u}-\ket{v}} \leq \epsilon$.

\paragraph{$\MIP^*$ and nonlocal games.}
We follow the notations of \cite{neexp} for two-player one-round $\MIP^\ast$ protocols.
For some language $L$, such a protocol is an interaction between a verifier and two noncommunicating provers, who can share arbitrary entanglement.
Given an instance $\ipt$, the verifier first samples a pair of random questions and sends them to the provers, receives
a pair of answers from the provers, and decides whether to accept $\ipt \in L$ or reject it based on the questions and answers.
For $\MIP^*$, the verifier is described by two algorithms: a randomized sampler $Q$ and a deterministic verification algorithm $V$.
Given $\ipt$, the protocol is denoted by $P(\ipt)$.
To execute the protocol, the verifier samples two questions $(x, y) \sim Q(\ipt)$ and gives question $x$
to Alice and question $y$ to Bob.
After receiving answers $a$ and $b$, the verifier accepts if $V(\ipt, x, y, a, b) = 1$.
Let $n = \abs{\ipt}$. Then the question length, answer length, sampling time and verification time of $P(\ipt)$
are denoted by $q(n)$, $a(n)$, $t_Q(n)$ and $t_V(n)$.
When analyzing answer reduction transformation, we break $V$ into two phases. In the first phase, $V$ computes
a predicate circuit $C^{\ipt}_{x,y}$ to check the answers. In the second phase, $V$ runs $C^{\ipt}_{x,y}$ on answers $a,b$ and outputs the output of $C^{\ipt}_{x,y}$. We call the size of $C^{\ipt}_{x,y}$ the \emph{decision complexity} of $P$ and denote it by $d_V(n)$.

Two-player one-round $\MIP^\ast$ protocols are also nonlocal games.
We follow the notations of \cite{re} for nonlocal games.
\begin{definition}[Two-player one-round games]
    A two-player one-round game $G$ is specified by a tuple $(\calX, \calY, \calA, \calB, \mu, V)$ where
    \begin{itemize}
        \item $\calX$ and $\calY$ are finite sets, called the \emph{question sets},
        \item $\calA$ and $\calB$ are finite sets, called the \emph{answer sets},
        \item $\mu$ is a probability distribution over $\calX \times \calY$, called the \emph{question distribution}, and
        \item $V: \calX \times \calY \times \calA \times \calB \to \set{0,1}$ is a function, called the \emph{decision predicate}.
    \end{itemize}
\end{definition}
In the context of $\MIP^*$, $\mu$ is the distribution that $Q(\ipt)$ samples from, and the predicate $V$ has
$\ipt$ hardcoded into it, so $V(x,y,a,b) = V(\ipt, x, y, a, b)$ for all $(x,y,a,b) \in \calX \times \calY \times \calA \times \calB$.

A special type of two-player one-round games is the synchronous game. In such a game, the question sets of the two players are the same, which we denote as $\calX$, the question distribution $\mu$ is a symmetric probability distribution over $\calX\times\calX$, and the decision predicate $V$ is a symmetric function satisfying that if both players are asked the same question $x$, $V(x,x,a,b)=1$ only if $a=b$. Furthermore, we call the synchronous game $\alpha$-synchronous if there exists $\alpha\in(0,1)$ such that $\mu(x,x)\geq\alpha\sum_y\mu(x,y)$ for every $x\in\calX$.

Next, we formally define the strategies for nonlocal games and $\MIP^*$ protocols and the values they achieve.
\begin{definition}[Tensor-product strategies and correlations]
    A tensor-product strategy $S$ of a nonlocal game $G = (\calX, \calY, \calA, \calB, \mu, V)$ is a tuple 
    $(\ket{\psi}, A, B)$ where
    \begin{itemize}
        \item a bipartite quantum state $\ket{\psi} \in \calH_A \x \calH_B$ for finite dimensional complex Hilbert
spaces $\calH_A$ and $\calH_B$,
        \item $A$ is a set $\set{A^x}$ such that for every $x \in \calX$, $A^x = \set{A^x_a \mid a \in \calA}$ is 
        a projective measurement over $\calH_A$, and 
        \item $B$ is a set $\set{B^y}$ such that for every $y \in \calY$, $B^y = \set{B^y_b \mid b \in \calB}$ is 
        a projective measurement over $\calH_B$.
    \end{itemize}
    In the nonlocal game, given questions $x$ and $y$, the probability distribution for the answers $a$ and $b$ is given by
    \begin{equation}
        P(a,b\mid x,y) = \bra{\psi} A^x_a \x B^y_b \ket{\psi}.
    \end{equation}
    We call the set of question-conditioned distributions $\{P(a,b\mid x,y)\}_{a,b,x,y}$ the \emph{correlation} of the tensor-product strategy $S$.
    
\label{def:QStrategy}
\end{definition}
We will not discuss commuting operator strategies in our work, so we also refer to tensor-product strategies as quantum strategies.
Note that we focus on projective measurements in the definition above. This is because 
when we prove the quantum soundness of an $\MIP^\ast$ protocol, we can assume the measurements are projective by Naimark's Dilation theorem \cite[Theorem 5.1]{re}.

\begin{definition}[Tensor product value]
    The tensor-product value of a tensor product strategy $S = (\ket{\psi}, A ,B)$ for a nonlocal game
    $G = (\calX, \calY, \calA, \calB, \mu, V)$ is defined as
    \begin{align*}
        \val^\ast(G,S) = \sum_{x,y,a,b} \mu(x,y)V(x,y,a,b) \bra{\psi} A^x_a \x B^y_b \ket{\psi}.
    \end{align*}
    For $v \in [0,1]$ we say that the strategy passes or wins $G$ with probability $v$ if $\val^\ast(G,S)\geq v$.
    The quantum value or tensor product value of $G$ is defined as
    \begin{align*}
        \val^*(G) = \sup_{S} \val^\ast(G,S)
    \end{align*}
    where the supremum is taken over all tensor product strategies $S$ for $G$.
\end{definition}
In other words, the quantum value of Alice and Bob's strategy is the probability for the verifier to accept Alice and Bob's answers over the distribution of the questions.
We say a strategy $S$ that achieves quantum value 1 is a \emph{perfect quantum strategy},
and say the correlation produced by $S$ is a \emph{perfect correlation}.

There is
a special subclass of quantum strategies called \emph{projective, commuting and consistent} (PCC) strategies.
Since we only consider projective measurements in this work, we define commuting strategies and consistent strategies below.
\begin{definition}[PCC strategy]
    \label{def:pcc}
    We say a projective strategy $S = (\ket{\psi}, A ,B)$ for a nonlocal game
    $G = (\calX, \calY, \calA, \calB, \mu, V)$ is \emph{commuting} if for any $(x,y) \in \calX \times \calY$ such that
    $\mu(x,y) > 0$, $[A^x_a, B^y_b] = A^x_aB^y_b - B^y_bA^x_a = 0$ for all $a \in \calA$ and $b \in \calB$.
    Moreover, we say $S$ is \emph{consistent} if for all $x \in \calX, a \in \calA$ and $y \in \calY, b\in \calB$
    \begin{align*}
        A^x_a \x \1 \ket{\psi} = \1 \x A^x_a \ket{\psi} && 
         B^y_b \x \1 \ket{\psi} = \1 \x B^y_b \ket{\psi}.
    \end{align*}
\end{definition}
\begin{definition}
    We say a language $L$ is in $\MIP^*_{c,s}[q,a,t_Q,t_V,d_V]$, where $q,a,t_Q, t_V, d_V,c, s$ are functions of the instance size, if there exists a protocol $P = (Q,V)$ whose question length, answer length, sampling time, verification time, and decision complexity
    are upper bounded by $q, a, t_Q, t_V, d_V$ respectively such that:
    \begin{description}
        \item[Completeness] If $\ipt \in L$, there exists a quantum strategy $S=(\ket{\psi},M,N)$ for $P(\ipt)$ whose value is at least $c$.
        \item[Soundness] If $\ipt \notin L$, there is no quantum strategy whose value is larger than $s$ for $P(\ipt)$.
    \end{description}
\end{definition}

When using only the first four parameters, we use $t_V$ as an upper bound of $d_V$.
The decision complexity will be explicit in the analysis of answer reduction.
In this work, we focus on the case that $c=1$ and $s$ is a constant smaller than 1.
To simplify the notation, we sometimes use $V$ to represent a verifier.
If an $\MIP^*$ protocol has a perfect PCC  strategy, then the question-reduced and answer-reduced protocols also have a perfect PCC  strategy. We need this property to argue the completeness of $\MIP^*$ protocols are preserved through question reduction and answer reduction.

\begin{definition}[Zero knowledge $\MIP^*$ {\cite[Definition 6.3]{coudron2019complexity}}]
    Let $P=(Q,V)$ be a two-prover one-round $\MIP^*$ protocol for a language $L$ with completeness $1$ and 
    soundness $s$.
    For any $\ipt \in L$, let $\Sim$ be a classical simulator running in a polynomial time of $\abs{\ipt}$,
    and let $\mathrm{D}$ be a distinguisher also running in a polynomial time of $\abs{\ipt}$.
    We say the protocol $P$ is
    \begin{description}   
        \item[Perfect zero-knowledge (PZK),] if there is a perfect quantum strategy $S = (\ket{\psi}, M, N)$ whose correlation is $\{P(a,b\mid x,y)\}_{x,y}$, such that for any $x\in\mathcal{X},y\in\mathcal{Y}$, there exists a simulator $\Sim$ that can sample from the distribution $P(a,b\mid x,y)$;
        \item[Statistical zero-knowledge (SZK),] if there is a perfect quantum strategy $S = (\ket{\psi}, M, N)$ whose correlation is $\{P(a,b\mid x,y)\}_{x,y}$, such that for any $x\in\mathcal{X},y\in\mathcal{Y}$, there exists a simulator $\Sim$ that can sample from a distribution whose statistical distance to $P(a,b\mid x,y)$ is smaller than any inverse polynomial in $\abs{\ipt}$;
        \item[Computational zero-knowledge (CZK),] if there is a perfect quantum strategy $S = (\ket{\psi}, M, N)$ whose correlation is $\{P(a,b\mid x,y)\}_{x,y}$, such that for any $x\in\mathcal{X},y\in\mathcal{Y}$, there exists a simulator $\Sim$ that can sample from a distribution, which cannot be distinguished from $P(a,b\mid x,y)$ by any distinguisher $\mathrm{D}$.
    \end{description}
    The class $\pzkMIP^*_{c,s}$ (resp. $\szkMIP^*_{c,s}$ and $\czkMIP^*_{c,s}$) is the class of languages with a
    perfect zero-knowledge (resp. statistical zero-knowledge and computational zero-knowledge) two-prover one-round $\MIP^*$ protocol
    with completeness $c$ and soundness $s$.
\end{definition}

In the above zero knowledge definition, we account for dishonest verifiers. A dishonest verifier may deviate from the $\MIP^*$ protocol and seek more information about the proof of the provers, where instead of sampling questions $(x,y)$ from the desired distribution $\mu$ in the $\MIP^*$ protocol, they sample questions from a different distribution that can be sampled in a polynomial time of $|\ipt|$.
Then, PZK implies that no matter how the dishonest verifier changes the distribution, there exists an efficient simulator to produce the same answer distribution conditioned on the questions. Thus, the verifier still cannot learn anything about the proof. Statistical (Computational) zero-knowledge guarantees that the verifier cannot distinguish the true distribution from a simulated distribution without the proof, so the dishonest verifier can learn at most negligible portion of the proof. The validity of the above definition against a dishonest verifier has been discussed in \cite[Definition 6.3]{coudron2019complexity} and \cite[Theorem 9.1]{mastel2024two}.

\paragraph{Binary constraint system games.}
In our results, we shall utilize a special type of nonlocal games inspired by the binary constraint systems (BCS).
Specifically, a BCS consists of $n$ binary variables and $m$ constraints. For later convenience, we use signed variables, where each variable $v_j$ for $j\in[n]$ takes values in $\{\pm1\}$. 
In the following discussions, we shall refer to a variable by its index $j\in[n]$. 
In the BCS, each constraint is a Boolean function over a subset of variables, and we denote the maximum number of a constraint's composing variables as $C$. For the $i$'th constraint, we denote the subset of variables within it as $V_i$. 
Additionally, for the $i$'th constraint, we denote the set of satisfying assignments over its composing variables as $C_i$, which is a subset of $\{\pm1\}^{V_i}$. 
From another viewpoint, since every Boolean function can be represented as a multi-linear polynomial over the real field, we can consider each constraint as defined by a polynomial $c_i$ over the set of composing variables $V_i$, and the constraint is represented as a function in the form of $c_i(V_i)=\pm1$.
BCS is widely used in classical theoretical computer science. For instance, every 3-SAT instance defines a BCS.

Once a BCS is given, we can define an associated nonlocal game~\cite{cleve2014characterization} . In a constraint-constraint BCS nonlocal game, the two players, Alice and Bob, are each given a constraint labelled by $i\in[m]$ and $j\in[m]$, respectively. The players need to output a satisfying assignment to the constraint they each receive. Furthermore, for every variable that shows up in both constraints, i.e., $\forall k\in V_i\cap V_j$, the values assigned to it must be the same.
In a constraint-variable BCS nonlocal game, one of the players, say Alice, is given a constraint $i\in[m]$, and Bob is given a variable $j\in V_i$. In this game, Alice needs to output a satisfying assignment to the constraint, and her assignment to the variable $j$ must take the same value as Bob's.

As a famous BCS example, the Mermin-Peres magic square~\cite{mermin1990simple,peres1990incompatible} is given as:
\begin{equation}
\begin{aligned}
    v_1v_2v_3 &= 1, \\
v_4v_5v_6 &= 1, \\
v_7v_8v_9 &= 1, \\
v_1v_4v_7 &= 1, \\
v_2v_5v_8 &= 1, \\
v_3v_6v_9 &= -1.
\end{aligned}
\end{equation}
This BCS contains nine variables, six constraints, and each constraint contains three variables.
All the constraints in this BCS cannot be simultaneously satisfied. Consequently, the associated BCS game cannot be won perfectly via classical strategies. Nevertheless, by adopting tensor-product strategies in \cref{def:QStrategy}, the nonlocal players can win the game perfectly.
In the perfect-winning strategy, the players needs to share a maximally entangled state, $\ket{\phip}$.
Moreover, the set of quantum measurements of each player can be seen as a operator-valued satisfying assignment to the BCS. That is, for the observables associated with variables measured by Alice (resp. Bob), $\{A^x\}$ (resp. $\{B^y\}$), with respect to each constraint $c_i(V_i)=\pm1$, we have $c_i(\{A^x\}_{x\in V_i})=\pm\1$ (resp. $c_i(\{B^y\}_{y\in V_i})=\pm\1$), where $\1$ is the identity operator.


\section{Constraint-constraint game to symmetrized constraint-variable game}
The PZK protocols from \cite{mastel2024two} can be viewed as 
constraint-constraint BCS games where each prover gets a uniformly random constraint.
To compress such protocols while preserving perfect completeness, they need to be converted to constraint-variable games. 
The purpose of this conversion is to guarantee the PCC condition for the perfect strategy. This is vital for answer reduction, as answer reduction can only be applied to an oracularized protocol, and oracularization preserves completeness only if the original game satisfies the PCC condition \cite{natarajan2023quantum}.
Note that the zero-knowledge protocol from \cite{mastel2024two}, which uses constraint-constraint games, does not satisfy the PCC condition in the complete case.
Moreover, since question reduction and oracularization preserve the PCC condition, we will use this conversion to make sure the original game for question reduction and answer reduction satisfies the PCC condition.

For a constraint-constraint BCS game with $m$ constraints and $n$ variables, in its corresponding constraint-variable BCS game, Alice gets a random constraint $i \in [m]$ and Bob gets a random variable $j\in V_i$ in Alice's constraint. 
Nevertheless, another subtlety is that such a constraint-variable game is not synchronous, which is required for the soundness analysis for a technical reason. For this purpose, we apply a symmetrization processing step. That is, instead of keeping asking one particular player constraint questions and the other variable questions, the verifier randomly specifies Alice or Bob to answer the constraint question and asks the other to answer an associated variable question. In addition, the verifier asks consistency-check questions randomly in some rounds, where the two players are asked the same constraint question or the same variable question. In this way, the question sets of both players are the same, which are defined by the union of the set of constraints and the set of variables in the underlying BCS. We call such a game a symmetrized constraint-variable BCS game.

Here, we specify the question distribution we use in the symmetrized constraint-variable game $G^{cv}$. 
First, the referee samples a constraint $i\in [m]$ uniformly at random. The referee then samples a uniformly random variable $j$ from constraint $i$. The referee asks one of the following question pairs with uniform probability.
\begin{itemize}
    \item The referee asks Alice and Bob question $i$.
    \item The referee asks Alice and Bob question $j$.
    \item The referee asks Alice question $i$ and Bob question $j$.
    \item The referee asks Alice question $j$ and Bob question $i$.
\end{itemize}


Our first result is the following theorem.

\begin{theorem}
    \label{thm:cc_to_cv}
    Let $G$ be a constraint-constraint BCS game with $n$ variables, $m$ constraints and at most $C$ variables in each constraint, and $G^{cv}$ be the corresponding symmetrized constraint-variable BCS game.
    \begin{description}
        \item[Completeness] If $G$ has a perfect finite-dimensional quantum strategy, $G^{cv}$ has a perfect PCC quantum strategy whose shared state is maximally entangled.
        \item[Soundness] If $G^{cv}$ has a quantum strategy with winning probability $1-\epsilon$, then $G$ has a quantum strategy achieving quantum value at least $1 - \poly(C, \epsilon)$.
        \item[Zero-knowledge] If the corresponding $\MIP^*$ protocol of $G$ has perfect completeness and is perfect zero-knowledge (resp. statistical zero-knowledge and computational zero-knowledge), so is the corresponding protocol of $G^{cv}$.
    \end{description}
\end{theorem}

\begin{proof}
For convenience of our proof, we denote $\tilde{G}^{cv}$ as the constraint-variable game before the symmetrization step.
In the constraint-constraint BCS game, $G$, let Alice and Bob's quantum strategy be $S = (\ket{\psi}, M, N)$, where
\begin{enumerate}
    \item $\ket{\psi}$ is a finite or countably-infinite dimensional entangled state;
    \item Alice: for a constraint question $i\in[m]$, her measurement is given by projectors $\{M^i_a\}_a$, with $a\in C_i$ giving a satisfying assignment to the constraint;
    \item Bob: for a constraint question $j\in[m]$, his measurement is given by projectors $\{N^i_b\}_b$, with $b\in C_i$ giving a satisfying assignment to the constraint.
\end{enumerate}

\paragraph{Completeness.} 
Suppose the strategy $S$ can win the constraint-constraint game $G$ perfectly. On Alice's side, denote $M^{i,k}=\sum_{a \in C_i} (-1)^{a_k} M^i_a$ with $a_k=\pm1$, representing the observable for the variable labelled by $k$ in the $i$'th constraint, and we take a similar notation on Bob's side. We decompose the entangled state $\ket{\psi}$ on the Schmidt basis:
\begin{equation}\label{eq:Schmidt}
    \ket{\psi}=\sum_{i=1}^{\infty}\alpha_i\ket{\phi_i}\ket{\varphi_i},
\end{equation}
where $\{\ket{\phi_i}\}_i$ and $\{\ket{\varphi_i}\}_i$ form orthonormal bases of $\mathcal{H}$, and $\alpha_i\geq0,\forall i$, and $\sum_{i=1}^\infty\alpha_i^2=1$. In the following discussions, we restrict the Hilbert space to be the support of $\ket{\psi}$, i.e., the spaces of the subsystems are given by $\mathcal{H}=\mathrm{span}(\{\ket{\phi_i}\}_{i:\alpha_i>0})$, which is isomorphic to $\mathrm{span}(\{\ket{\psi_i}\}_{i:\alpha_i>0})$. In addition, we have the following lemma:

\begin{lemma}[Lemma~2, \cite{cleve2014characterization}]
  Consider the quantum state $\ket{\psi}\in\mathcal{H}\x\mathcal{H}$ given in Eq.~\eqref{eq:Schmidt}, and Hermitian operators $B,C_1,C_2$ acting on $\mathcal{H}$, such that $-\1 \leq C_1,C_2,B\leq \1$. If $\bra{\psi}B\x C_1\ket{\psi}=\bra{\psi}B\x C_2\ket{\psi}=1$, then $C_1=C_2$ \footnote{If some eigenvectors of these operators reside outside $\mathcal{H}$, they do not enter the calculation and their effects are also not observable.}.
\label{lemma:OperatorId}
\end{lemma}

Based on Eq.~\eqref{eq:Schmidt} and Lemma~\ref{lemma:OperatorId}, we prove completeness with several steps. 

\begin{enumerate}
    \item Consider $k\in V_i\cap V_j$, and the following two cases: (1) Alice receives the $i$'th constraint, Bob receives the $i$'th constraint, and (2) Alice receives the $i$'th constraint, Bob receives the $j$'th constraint. Then,
    \begin{align*}
        \bra{\psi}M^{i,k}\x N^{i,k}\ket{\psi}&=1, \\
        \bra{\psi}M^{i,k}\x N^{j,k}\ket{\psi}&=1.
    \end{align*}
    By applying Lemma~\ref{lemma:OperatorId}, we have $N^{i,k}=N^{j,k}$. It follows that Bob's effective observables for the $k$'th variable in the two cases are the same. 
    Hence for any $i$ such that $k\in V_i$, we can denote $N^{i,k}\equiv B^{k}$.
    
    \item Using a similar argument, Alice's effective observables for the same variable in various constraints are the same, which we denote as $\{A^{k}\}_k$.
    
    \item As Alice's assignments need to satisfy the constraints, for the $i$'th constraint,
    \begin{align*}
        \bra{\psi}c_i(M^{i,k\in V_i})\otimes \1\ket{\psi}=1,
    \end{align*}
    where $c_i$ represents the multilinear polynomial function that defines the $i$'th constraint in the BCS.
    By applying Lemma~\ref{lemma:OperatorId}, $c_i(M^{i,k\in V_i})\equiv c_i(A^{k\in V_i})=\1$. Therefore, $\{A^k\}_k$ gives an operator-valued satisfying assignment to the BCS. Similarly, $\{B^k\}_k$ also gives an operator-valued satisfying assignment to the BCS.
    \item For every $k$, $A^k$ and $B^{k}$ satisfy
    \begin{align*}
        \bra{\psi}A^k \x B^k\ket{\psi}=1,
    \end{align*}
    hence $A^k\x \1\ket{\psi}=\1\x B^k\ket{\psi}$, given by
    \begin{align*}
        \sum_{i=1}^\infty\alpha_i(A^k\ket{\phi_i})\ket{\varphi_i}=\sum_{i=1}^\infty\alpha_i\ket{\phi_i}(B^k\ket{\varphi_i}).
    \end{align*}
    Both sides are also Schmidt decompositions of the state $\ket{\psi}$. Note that the Schmidt decomposition is unique up to a unitary operator acting on a subspace spanned by vectors of the same Schmidt coefficient. Therefore, for every subspace, we have $\mathcal{H}_{\alpha}=\mathrm{span}(\ket{\phi_i}:\alpha_i\equiv\alpha)$, $A^k\mathcal{H}_{\alpha}=\mathcal{H}_{\alpha},\forall k$. A similar argument can be made on Bob's side. Then, the operators restricted to the subspace $\mathcal{H}_{\alpha}$ also give a quantum satisfying assignment to the BCS, and the effective underlying state is the maximally entangled state $\ket{\phi^+}$. Since
    \begin{align*}
        \bra{\phip}M^{i,k}\otimes N^{i,k}\ket{\phip}=\bra{\phip}M^{i,k}(N^{i,k})^\top\otimes \1\ket{\phip}=1,
    \end{align*}
    and the operators $M^{i,k},N^{i,k}$ have eigenvalues $\pm1$, we have $M^{i,k}=N^{i,k}$.

    \item Next, we show that there exists a perfect strategy that wins the constraint-variable game $\tilde{G}^{cv}$ perfectly. For this purpose, Alice and Bob share the maximally entangled state $\ket{\phip}$, Alice applies the measurements given by projectors $\{M^i_a\}_{a}$ for each constraint $i$, and Bob applies the measurements given by observables $B^k$ for each variable $k$, where $B^k$ are observables given by $N^{i,k}$ in step~1. 
    
    \item Finally, in the symmetrized constraint-variable game $G^{cv}$, the additional questions other than those asked in $\tilde{G}^{cv}$ are the consistency check questions. Because the players share $\ket{\phip}$ and that the observables they measure are the same, they will output identical measurement outcomes. This completes the proof.
\end{enumerate}


\paragraph{Soundness.}
Suppose a quantum strategy $(\ket{\psi}, M, N)$ wins $G^{cv}$ with probability at least $1 -\epsilon$. We first round it to an EPR quantum strategy that wins $G^{cv}$ with sufficiently high probability. We take advantage of the following theorem for the proof:

\begin{theorem}[{\cite[Theorem 0.1]{marrakchi2023almost}}]
    Let $\varepsilon,\alpha\in(0,1)$, and let $G$ be an $\alpha$-synchronous game with a PCC strategy with value $1-\epsilon$. Then, there exists an EPR strategy for $G$ with value at least $1-c(\epsilon/\alpha)^{\frac{1}{4}}$, with $c$ being a universal positive constant.
\label{thm:rounding}
\end{theorem}

To apply this theorem to our case, we first examine the question distribution of $G^{cv}$. By construction, the question distribution is symmetric, thus we shall first fix the entry of Alice's question and examine the probability of possible question pairs.

\begin{enumerate}
    \item Case 1: Consider that Alice is asked a constraint question $i\in[m]$. In this case, Bob is asked either the same constraint question or a variable question $j\in V_i$. The first event is chosen with probability $\mu(i,i)=\frac{1}{4}\frac{1}{m}$; the event $(i,j)$ with $j\in V_i$ is chosen with probability $\mu(i,j)=\frac{1}{4}\frac{1}{m}\frac{1}{|V_i|}$. Therefore, we have 
    \begin{align*}
        \sum_{\ell}\mu(i,\ell)=\frac{1}{4m}+\sum_{j\in V_i}\frac{1}{4m|V_i|}=\frac{1}{4m}+|V_i|\frac{1}{4m|V_i|} = \frac{1}{2m}.
    \end{align*}
    \item Case 2: Consider that Alice is asked a variable question $j\in[n]$. In this case, Bob is asked either the same variable question or a constraint question $i\in[m]$ such that $j\in V_i$. The first event is chosen with probability 
    \begin{align*}
        \mu(j,j)=\frac{1}{4}\frac{1}{m}\sum_{i:j\in V_i}\frac{1}{|V_i|};
    \end{align*}
    for the event $(j,i)$, this question is asked with probability $\frac{1}{4}$ when the referee samples the constraint-variable pair $(i,j)\in[m]\times[n]$. Thus, the event occurs with probability $\mu(j,i)=\frac{1}{4}\frac{1}{m}\frac{1}{|V_i|}$. 
    Therefore, we have
    \begin{align*}
        \sum_{\ell}\mu(j,\ell)=\frac{1}{4}\frac{1}{m}\sum_{i: j\in V_i}\frac{1}{|V_i|}+\frac{1}{4}\frac{1}{m}\sum_{i:j\in V_i}\frac{1}{|V_i|}
        .
    \end{align*}
\end{enumerate}
In case 1, we have
\begin{align*}
    \mu(i,i)\geq\frac{1}{2}\sum_{\ell}\mu(i,\ell),
\end{align*}
for every $i\in[m]$. Similarly, in case 2, we get
\begin{align*}
    \sum_i\mu(j,j)\geq\frac{1}{2}\sum_{\ell}\mu(j,\ell),
\end{align*}
for every $j\in[n]$.
We conclude that the symmetrized constraint-variable game is $\alpha$-synchronous with $\alpha=\frac{1}{2}$. Consequently, applying Theorem~\ref{thm:rounding}, there exists an EPR strategy winning $G^{cv}$ with probability at least $1 - \poly(\epsilon^{1/4})$, where $\ket{\phi_+}$ is a maximally entangled state (MES).
Recall that there are $m$ constraints, $n$ variables in the BCS and at most $C$ variables in the each constraint.
Then we want to show there is a strategy for the corresponding constraint-constraint BCS game $G$ with a winning probability at least $1 - \poly( C,\epsilon)$.


Starting with the EPR strategy for $G^{cv}$ denoted as $(\ket{\phip},M,N)$, let $\delta \coloneqq \poly(\epsilon^{1/4})$. For each $i \in [m]$ and $k \in V_i$, let $M^{i,k} = \sum_{a \in C_i} (-1)^{a_k} M^i_a$ and $N^k = N^k_0 - N^k_1$. Let $\delta_i$ be the probability that the player's strategy fails given that the referee sampled constraint $i$. Note that $\delta \geq \sum_i\frac{1}{m}\delta_i$. We have the following bound on Alice and Bob's strategy. 


\begin{align*}
    \bra{\phip} M^{i,k} \x N^k \ket{\phip} = 2\pr{\text{win}}{i,k} - 1\geq 1 - 2 |V_i|\delta_i\geq 1-2C\delta_i,
\end{align*}
or equivalently
\begin{align*}
    \norm{M^{i,k} \x \id \ket{\phip} - N^k \x \id \ket{\phip} }^2 \leq 4 C\delta_i \text{ and }
    M^{i,k} \x \id \ket{\phip} \appd{2C\sqrt{\delta_i}} N^k \x \id \ket{\phip}.
\end{align*}
Hence for any $i, j \in [m]$ such that $k \in V_i \cap V_j$ and $b = 0,1$,
\begin{align*}
    M^{i,k}_b \x \id \ket{\phip} \appd{C\sqrt{\delta_i}} N^k_b \x \id \ket{\phip} \appd{C\sqrt{\delta_j}}
    M^{j,k}_b \x \id \ket{\phip}.
\end{align*}
Since $\ket{\phip}$ is an MES and $M^{i,k}$ and $N^k$ are Hermitian, we also have
\begin{align*}
    \id \x M^{i,k}_b \ket{\phip} \appd{C\sqrt{\delta_i}} \id \x N^k_b  \ket{\phip} \appd{C\sqrt{\delta_j}}
    \id \x M^{j,k}_b \ket{\phip}.
\end{align*}
Then,
\begin{align*}
    M^{i,k}_b \x M^{j,k}_b \ket{\phip} \appd{C\sqrt{\delta_j}} M^{i,k}_b \x N^{k}_b \ket{\phip}
    \appd{C\sqrt{\delta_i}} N^k_b \x N^{k}_b \ket{\phip} = N^k_b \x \id \ket{\phip},
\end{align*}
where we use the projector property of $(N^k_b)^2=N^k_b$ and that $A\otimes B\ket{\phip}=AB^{\mathrm{T}}\ket{\phip}$ in the last equation.
By applying the triangle inequality,
\begin{equation}\label{eq:substitute}
    \norm{M^{i,k}_b \x M^{j,k}_b \ket{\phip} - N^k_b \x \id \ket{\phip}} \leq C\left(\sqrt{\delta_i}+\sqrt{\delta_j} \right)
\end{equation}

Suppose there are two variables labelled by $k_1$ and $k_2$ where $k_1,k_2\in V_i\cap V_j$. The projector for the measurement results $b_1$ and $b_2$ on Alice's side is given by
\begin{equation}
      \sum_{\substack{\vec{a}\in C_i: \\a_{k_1}=b_1 \\a_{k_2}=b_2}}M^i_{\vec{a}} 
    = \sum_{\substack{\vec{a}_1\in C_i:\\
    a_{1_{k_1}}=b_1}}M^i_{\vec{a}_1} 
      \sum_{\substack{\vec{a}_2\in C_i:\\ 
    a_{2_{k_2}}=b_2}}M^i_{\vec{a}_2} =M^{i,k_1}_{b_1}M^{i,k_2}_{b_2},
\end{equation}
where in the first equation, we use the orthogonality between projectors $M^{i}_{\vec{a}_1}M^{i}_{\vec{a}_2}=0$ if $\vec{a}_1\neq\vec{a}_2$. Hence the winning condition is given by
\begin{align*}
    \pr{\text{win}_{k_1} \land \text{win}_{k_2}}{i,j}  &= \sum_{b_1,b_2=0,1}
    \sum_{\vec{a}\in C_i} \sum_{\substack{\vec{a'}\in C_j \\a_{k_1}=a'_{k_1}=b_1 \\a_{k_2}=a'_{k_2}=b_2}}\bra{\phip}M^i_{\vec{a}}\otimes M^j_{\vec{a}'}\ket{\phip} \\
    &= \sum_{b_1=0}^{1}\sum_{b_2=0}^{1}\bra{\phip} (M^{i,k_1}_{b_1} \x M^{j,k_1}_{b_1}) (M^{i,k_2}_{b_2} \x M^{j,k_2}_{b_2}) \ket{\phip},
\end{align*}
where $\text{win}_{k_1}$ is the event that the players give satisfying assignments agreeing on the assignment to the variable labelled by $k_1$.
In general, for $|V_i\cap V_j|=\ell\leq C$, where $k_1,\cdots,k_\ell\in V_i\cap V_j$, we have
\begin{equation}
    \sum_{\substack{\vec{a}\in C_i: \\a_{k_t}=b_t,\,t\in[\ell]}}M^i_{\vec{a}}= \sum_{\substack{\vec{a}_1\in C_i: \\a_{1_{k_1}}=b_1}}M^i_{\vec{a}_1} \sum_{\substack{\vec{a}_2\in C_i: \\ a_{2_{k_2}}=b_2}}M^i_{\vec{a}_2} \cdots \sum_{\substack{\vec{a}_\ell\in C_i: \\ a_{\ell_{k_\ell}}=b_l}}M^i_{\vec{a}_\ell}
    =\prod_{t=1}^\ell M^{i,k_t}_{b_t},
\end{equation}
and
\begin{align*}
    \pr{\land_{t=1}^\ell\text{win}_{k_t}}{i,j} 
    &=\sum_{b_1,\ldots,b_\ell = 0,1}
    \sum_{\vec{a}\in C_i} \sum_{\substack{\vec{a}'\in C_j \\a_{k_t}=a'_{k_t}=b_t,\,t\in[\ell]}}\bra{\phip}M^i_{\vec{a}}\x M^j_{\vec{a}'}\ket{\phip} \\
    &= \sum_{b_1=0}^{1}\cdots\sum_{b_\ell=0}^{1}\bra{\phip} (M^{i,k_1}_{b_1} \x M^{j,k_1}_{b_1}) \cdots (M^{i,k_\ell}_{b_\ell} \x M^{j,k_\ell}_{b_\ell}) \ket{\phip} \\
    &= \bra{\phip} \prod_{t=1}^{\ell}(\sum_{b_t=0}^{1}M^{i,k_t}_{b_t} \x M^{j,k_t}_{b_t})\ket{\phip}.
\end{align*}
Note that the product and the summation can be exchanged due to the orthogonality between different projectors $M^i_{\vec{a}}$.
Together with Eq.~\eqref{eq:substitute}, in the constraint-constraint game, 
\begin{align*}
    \pr{\land_{t=1}^\ell\text{win}_{k_t}}{i,j}
    = &\bra{\phip} \prod_{t=1}^{\ell}(\sum_{b_t=0}^{1}M^{i,k_t}_{b_t} \x M^{j,k_t}_{b_t})\ket{\phip} \\
    \appd{C\sqrt{\delta_i}+C\sqrt{\delta_j} } &\bra{\phip} \prod_{t=1}^{\ell-1}(\sum_{b_t=0}^{1}M^{i,k_t}_{b_t} \x M^{j,k_t}_{b_t})(N^{k_\ell}_0+N^{k_\ell}_1)\x \1\ket{\phip} \\
    &\cdots \\
    \appd{C\sqrt{\delta_i}+C\sqrt{\delta_j}} &\bra{\phip} \prod_{t=1}^{\ell}(N^{k_t}_0+N^{k_t}_1)\x \1\ket{\phip} \\
    =&\bra{\phip} \1\x \1\ket{\phip} \\
    =&1.
\end{align*}
There are $\ell$ steps in this deduction. By applying the triangle inequality, we have
\begin{align*}
    1-\pr{\land_{t=1}^\ell\text{win}_{k_t}}{i,j}
    \leq \ell C\left(\sqrt{\delta_i}+\sqrt{\delta_j} \right).
\end{align*}
The constraint questions are chosen uniformly at random. By averaging over questions $i,j$, 
\begin{align*}
    1-P(\text{win}) &=1-\sum_{i,j}P(i,j)\pr{\land_{t=1}^C\text{win}_{k_t}}{i,j} \\
    &\leq lC\sum_{i,j}P(i,j)(\sqrt{\delta_i}+\sqrt{\delta_j}) \\
    &= lC\frac{1}{m^2}2m\sum_i\sqrt{\delta_i} \\
    &\leq 2lC\sqrt{\delta} \\
    &=\poly( C,\epsilon),
\end{align*}
where we apply the following proposition in the final step:
\begin{proposition}
    Given $\delta_i\in(0,1),i\in[m]$ and $\sum_{i=1}^{m}\delta_i/m\leq\delta$, $\sum_{i=1}^{m}\sqrt{\delta_i}\leq m\sqrt{\delta}$.
\end{proposition}
\begin{proof}
    Denote $S=\sum_{i=1}^{m}\sqrt{\delta_i}$, we have
    \begin{align*}
        S^2=\sum_{i=1}^{m}\delta_i+\sum_{i=1}^{m}\sum_{j\neq i}\sqrt{\delta_i\delta_j}.
    \end{align*}
    For every $i,j\neq i$, $\sqrt{\delta_i\delta_j}\leq(\delta_i+\delta_j)/2$, hence
    \begin{align*}
        \sum_{i=1}^{m}\sum_{j\neq i}\sqrt{\delta_i\delta_j}\leq\frac{1}{2}\sum_{i=1}^{m}\sum_{j\neq i}(\delta_i+\delta_j).
    \end{align*}
    Therefore,
    \begin{align*}
        S^2&\leq\sum_{i=1}^{m}\delta_i+\frac{1}{2}\sum_{i=1}^{m}\sum_{j\neq i}(\delta_i+\delta_j) \\
        &=\sum_{i=1}^{m}\delta_i+\frac{1}{2}\sum_{i=1}^{m}\left[(m-2)\delta_i+\sum_{j=1}^{m}\delta_j\right] \\
        &\leq m\delta+\frac{1}{2}(m-2)m\delta+\frac{1}{2}m^2\delta \\
        &=m^2\delta,
    \end{align*}
    and hence $S=\sum_{i=1}^{m}\sqrt{\delta_i}\leq m\sqrt{\delta}$.
\end{proof}

Combining this argument with the rounding result using Theorem~\ref{thm:rounding}, if $\exists(\ket{\psi}, M, N)$ that wins the symmetrized constraint-variable game $G^{cv}$ with probability at least $1 -\epsilon$, then $(\ket{\phip}, M, M)$ wins the constraint-constraint game $G$ with probability at least $1 - \poly(C,\epsilon)$.

\paragraph{Zero-knowledge.}
Suppose there is a simulator $\Sim$ that can simulate a perfect correlation for $G$. 
Then a simulator $\Sim'$ for $G^{cv}$ can be constructed as follows. When getting constraint $i$
and a variable in it, $\Sim'$ randomly samples a constraint $j$, feeds $(i,j)$ into $\Sim$, and uses $\Sim$'s answer for constraint $i$ as its answer for constraint $i$ and $\Sim$'s assignment for the variable labelled by $k$ as its answer for $k$. 
The correlation produced by $\Sim'$ is also perfect.

If the correlation produced by $\Sim$ can win $G$ with probability $1-\epsilon$, then for each constraint $i$,
$\Sim$ can give a satisfying assignment with a probability of at least $1 - m^2\epsilon$. Hence the winning probability of $\Sim'$ constructed above is at least $1-m^2\epsilon$. When $\epsilon$ is inverse superpolynomial, so is $m^2\epsilon$, which concludes the statistical zero-knowledge part.

For computational zero-knowledge, suppose there is a distinguisher $D'$ that can distinguish any correlation produced by simulators from perfect correlations of $G'$, we can build a distinguisher $D$ for $G$.
When $D$ receives questions $(i,j)$ and their respective assignments, it selects a random variable from constraint $i$, feeds the assignment to constraint $i$ and the random variable to $D'$ and records the output of $D'$. If it is yes, $D$ repeats the same steps for constraint $j$ and output what $D'$ outputs; otherwise, $D$ outputs no. Then computational zero-knowledge follows from contrapositivity. 

\end{proof}





\section{Oracularization}
In the compression of $\MIP^*$ protocols, an important technique is the oracularization of nonlocal games. Roughly speaking, in the oracularization of a nonlocal game $G$, one player, who is called an oracle player,  receives a pair of questions $(x,y)$ that were originally sent to the two players in $G$, and the other player, who is called an isolated player, receives the question $x$ or $y$ that was originally sent to one player in $G$. The oracle player needs to generate answers $(a,b)$ to both questions satisfying the predicate of $G$, and the two players need to generate consistent answers to the same question of $x$ or $y$.
The key point of oracularization is to allow the oracle player to calculate a proof for the answers in the nonlocal game, which is necessary for answer reduction in the analysis of $\MIP^*$.

The details of oracularization are as follows. Consider a nonlocal game $G$ with two players $A$ and $B$, where the questions to the players are given by $x\in\mathcal{X},y\in\mathcal{Y}$, which are drawn with respect to the probability distribution $\mu$ over $\mathcal{X}\times\mathcal{Y}$, the players' answers are given by $a\in\mathcal{A},b\in\mathcal{B}$, and a decision predicate $V$, with $V(x,y,a,b)=1$ indicating winning the game. We define its associated oracularized nonlocal game $G^{\mathrm{OR}}$ as follows: 
First, define the type set as $\mathcal{T}^{\mathrm{OR}}=\{\mathrm{ORACLE},A,B\}$, where the elements are called types or roles. We define the type graph $\Graph^{\mathrm{OR}}$ as an undirected graph, which takes $\mathcal{T}^{\mathrm{OR}}$ as the vertex set and contains edges linking $\mathrm{ORACLE}$ and $A$, $\mathrm{ORACLE}$ and $B$, and self-loops. 
In each round of the game $G^{\mathsf{OR}}$, an edge $e$ in $\mathcal{G}^{\mathrm{OR}}$ (including self-loops) is randomly chosen, and a question pair for the game $G$, $(x,y)$, is drawn according to the distribution $\mu$. Then, each player is first specified with a role $t\in\mathcal{T}^{\mathrm{OR}}$, which are respectively the two end vertices of $e$ (when $e$ represents a self-loop, both players receive the same role). Based on the role,
\begin{enumerate}
    \item When a player is specified with the role $\mathrm{ORACLE}$, the player is further given the pair of questions $(x,y)$, who is required to generate answers $(a,b)$ satisfying $V(x,y,a,b)=1$;
    \item When a player receives the role $A$, the player is further given the question $x$, who is required to generate an answer $a\in\mathcal{A}$;
    \item When a player receives the role $B$, the player is further given the question $y$, who is required to generate an answer $b\in\mathcal{B}$.
\end{enumerate}
The players win the oracularized game if their answers are consistent and satisfy the requirement of the roles. Namely, when both players receive the role $\mathrm{ORACLE}$, their outputs are identical and satisfy $V(x,y,a,b)=1$; when one player receives the role $\mathrm{ORACLE}$ and another player receives the role $A$ (resp. $B$), the first player's answer satisfies $V(x,y,a,b)=1$, and the value $a$ (resp. $b$) is identical to the answer of the second player; when both players receive the role $A$ (resp. $B$), their outputs are identical.

As a remark, the oracularization in our work follows~\cite{natarajan2023quantum} and differs from the one in~\cite{re}, where the type graph does not contain an edge linking vertices $A$ and $B$. This design guarantees the oracularized game to be a projection game~\cite{dinur2015parallel}, which is relevant to choosing the parallel repetition method later.

\begin{definition}[\cite{dinur2015parallel}]
    Consider a nonlocal game $G$ with two players. The game is called a projection game if for any pair of questions $(x,y)$, any possible answer $b$ from the second player determines at most one valid answer $a=\pi_{xy}(b)$ for the first player.
\label{def:projection}
\end{definition}

In summary, by \cite[Theorem 9.1]{re}
if $P=(Q,V)$ is an $\MIP^*_{1,s}[q,a, t_Q, t_V]$ protocol for a language $L$ with constant soundness $s$, 
the oracularized protocol $P^{\OR}=(Q^{\OR},V^{\OR})$ is an $\MIP^*_{1,s'}[q,a, t_Q, t_V]$ protocol for $L$ with constant soundness $s'$. 
Moreover, if $\ipt \in L$ and there is a perfect PCC strategy for $P(\ipt)$, there is a perfect PCC strategy for $P^{\OR}(\ipt)$.


Here we establish the zero-knowledge property of oracularization against \emph{honest} verifiers, who never samples questions that is not in the support of $\mu^{\OR}$, such as two ORACLE type questions.
\begin{proposition}
    \label{prop:oracul}
    If $P = (Q,V)$ is a $\pzkMIP^*_{1.s}[q,a,\poly,\poly]$  (resp. $\szkMIP^*_{1.s}$ or $\czkMIP^*_{1.s}$) protocol against dishonest verifiers for a language $L$,
    the oracularized protocol $P^{\OR} = (Q^{\OR}, V^{\OR})$ is a $\pzkMIP^*_{1.s'}[q,a,\poly,\poly]$  (resp. $\szkMIP^*_{1.s'}$ or $\czkMIP^*_{1.s'}$) against honest verifiers protocol for $L$.
\end{proposition}
\begin{proof}
    If $S$ is a simulator for $P$, then a simulator $S'$ for $P^{\OR}$ works as follows.
    If the question is of the form $((x,y), (x,y))$, $((x,y), x)$, or $((x,y), y)$, where $(x,y) \sim Q(\ipt)$,
    then $S'$ runs $S$ on input $(x,y)$ to get answers $a,b$ and outputs $((a,b),(a,b))$, $((a,b), a)$, or $((a,b),b)$
    according to the question.
    If the question is of the form $(x,x)$ or $(y,y)$, $S'$ samples Bob's question $y'$ (Alice's question $x'$),
    then $S'$ runs $S$ on input $(x,y')$ ($(x',y)$) to get answers $(a,b)$ and outputs $(a,a)$ ( $(b,b)$).
    This establishes perfect zero-knowledge and statistical zero-knowledge.
    
    For computational zero-knowledge, note that a distinguisher $D$ for $P^{\OR}$ can be used to construct a distinguisher $D'$ for $P$ as follows. 
    On input $(x,y,a,b)$, $D'$ inputs $((x,y), x, (a,b), a)$ to $D$ and outputs what $D$ outputs. Note that indistinguishability is only worsened by a constant factor because questions of type $(\mathrm{ORACLE},A)$ is sampled with a constant probability.
    This contradicts the computational zero-knowledge property of $P$.
\end{proof}
In \cref{sec:all}, we use specific properties of the base protocol to show the oracularized protocol preserves zero-knowledge properties against \emph{dishonest verifiers}.
\section{Parallel repetition}

The idea of parallel repetition is to sample from the question distribution $\mu$ for $k$ times, give all the questions to the provers, and accept only if all their answers are accepted by the verifier $V$.
In this work, we shall utilize two versions of quantum parallel repetition, which don't have any restriction on the question distribution $\mu$. 
The first one is a general result for all nonlocal games.
\begin{theorem}[\cite{yuen2016parallel}]
    Let $G$ be a nonlocal game with two players, where the maximum quantum winning probability is given by $\mathrm{val}^*(G)=1-\varepsilon$. Then for any integer value $n>0$, the maximum probability to win $n$ games in parallel using quantum strategies is upper-bounded by
    \begin{equation}
        \mathrm{val}^*(G^{\x n})\leq c\frac{a\log{n}}{\varepsilon^{17}n^{1/4}},
    \end{equation}
    where $c$ is a universal constant, and $a$ is the bit-length of the players' answers in the base game $G$.
\label{thm:polyparallel}
\end{theorem}

This parallel repetition result depends on the answer length $a$. For projection games defined in Def.~\ref{def:projection}, one can obtain a parallel repetition result that is independent of the answer length.

\begin{theorem}[\cite{dinur2015parallel}]
    Let $G$ be a projection game with two players, where the maximum quantum winning probability is given by $\mathrm{val}^*(G)=1-\varepsilon$. Then for any integer value $n>0$, there exists universal constants $c_0,c_1>0$, such that the maximum probability to win $n$ games in parallel using quantum strategies is upper-bounded by
    \begin{equation}
        \mathrm{val}^*(G^{\x n})\leq (1-c_0\varepsilon^{c_1})^{n/2}.
    \end{equation}
    In addition, the constant $c_1$ can be made $c_1\leq 12$.
\label{thm:projectionparallel}
\end{theorem}

Beyond soundness,
it is not hard to see that if there is a perfect PCC \hft{(EPR)} strategy for a nonlocal game $G$, there is also a perfect PCC \hft{(EPR)} strategy
for the parallel-repeated game $G^{\x m}$, because we can run the PCC strategy for $G$ in parallel.
The verification time of $G^{\x m}$ is $m$-times of that of $G$.
The perfect zero-knowledge property of parallel repetition is proved in \cite{mastel2024two},
so we establish the statistical and computational zero-knowledge property of parallel repetition.
\begin{proposition}
    \label{prop:parallel_repeat}
    If $P = (Q,V)$ is a $\pzkMIP^*_{1,s}[q,a,\poly,\poly]$  (resp. $\szkMIP^*_{1,s}$ or $\czkMIP^*_{1,s}$) protocol for a language $L$,
    then for $m = \poly(\abs{\ipt})$ with a sufficiently high polynomial degree,
    the $m$-times parallel repeated protocol $P^{\x m} = (Q^{\x m}, V^{\x m})$
    is a $\pzkMIP^*_{1,s'}[mq,ma,\poly,\poly]$  (resp. $\szkMIP^*_{1,s'}$ or $\czkMIP^*_{1,s'}$) protocol for $L$
    where the soundness $s'$ is polynomially reduced as in \cref{thm:polyparallel}. Furthermore, if $P = (Q,V)$ satisfies the projection game condition, the soundness $s'$ can be exponentially reduced as in
    \cref{thm:projectionparallel}.
\end{proposition}
\begin{proof}
    The perfect zero-knowledge property is proved in \cite[Proposition 9.2]{mastel2024two}.
    Suppose a simulator can sample a quantum distribution $p$ whose statistical distance to the perfect correlation is smaller than any inverse polynomial, then a new simulator $S'$, which runs $S$ $m$
    times independently, can sample from the distribution $p^{\otimes m}$ whose statistical distance to the ideal correlation is also smaller than any inverse polynomial, because $m$ is a polynomial of $\abs{\ipt}$.
    This proves the statistical zero-knowledge property.

    For the computational zero-knowledge property, suppose that there is a distinguisher $D$ for $V^{\otimes m}$. Then for any simulator $S$ of $V$, we can run it $m$ times and feed the outputs to $D$ to distinguish its output distribution from the ideal distribution.
    This completes the proof.
\end{proof}
\section{Question reduction preserves ZK}
\label{sec:qr}
The question reduction technique is first introduced in \cite{neexp} and improved and streamlined in \cite{re,natarajan2023quantum}.
In this section, we prove the zero-knowledge property of the streamlined version in \cite{natarajan2023quantum}.

Suppose $P(\ipt)$ is an instance of an $\MIP^*$ protocol, in which $Q(\ipt)$ samples questions $x$ and $y$ and expects answers $a$ and $b$ from the provers. 
The goal of question reduction is to reduce the length of the questions exponentially.
The central idea is to let the provers sample the questions $x$ and $y$ and prove to the verifier that they have sampled the questions from the right distribution.
This is achieved by executing the Pauli basis test, detailed in \cite[Section 7.3]{re}, to certify that the provers share $\abs{x}$ EPR pairs and that they do Pauli measurements on the shared state, and that
when a special question $\mathsf{Intro}$ (short for \emph{Introspection}) is sampled with a constant probability, they first measure in the Pauli Z basis and compute the questions $x$ and $y$ honestly.
Therefore, the honest provers will measure their shared EPR pairs in the Z basis to sample the questions when they get the question $\mathsf{Intro}$, and in the X or Z basis to answer the verifier's other checks.
In this way, the question length is reduced from $q$ to $\log(q)$, because the verifier doesn't need to send $x$ and $y$ anymore.
The question reduction of \cite{re} also preserves the perfect completeness and the PCC property of the completeness,
but the soundness will increase from a constant to $1 - O(1/\polylog(n))$.
To get around this issue, in the streamlined question reduction of \cite{natarajan2023quantum}, the question-reduced verifier is 
first oracularized to make the nonlocal game a projection game.
Then parallel repetition of \cite{dinur2015parallel} can be applied to keep the question length at $\polylog(n)$ but restore the constant soundness.

In summary, from the \cite[Theorem 50]{natarajan2023quantum}, we have the following parameters for the streamlined question reduction.
\begin{proposition}
    \label{prop:qr_summary}
    Let $P=(Q,V)$ be an $\MIP^*_{1, 1/2}[q, a, t_Q, t_V]$ protocol for a language $L$, and let $P^{QR}=(Q^{QR},V^{QR})$ be the question-reduced protocol obtained after applying the question reduction technique of \cite{re} to the protocol $P$. Then $P^{QR}$ is an $\MIP^*_{1, 1/2}[q', a', t'_Q, t'_V]$
protocol for $L$, with
\begin{description}
    \item[PCC property:] If $\ipt \in L$ and there is a perfect PCC strategy for $P(\ipt)$, then there is a perfect PCC strategy for $P(\ipt)^{QR}$,
    \item[Question length:] $q'(n) = \polylog(q(n))$,
    \item[Answer length:] $a'(n) = \polylog(q(n))(q(n)+a(n) + O(1))$,
    \item[Sampling time:] $t_Q'(n) = \polylog(q(n))$, and
    \item[Verification time:] $t_V'(n) = \polylog(q(n))(t_V(n) + q(n) + O(\log(n)))$.
\end{description}
\end{proposition}

Below we establish the zero-knowledge property of question reduction.

\begin{proposition}
    \label{prop:que_redu}
    If $P = (Q,V)$ is a $\pzkMIP^*_{1,1/2}[q,a,\poly,\poly]$  (resp. $\szkMIP^*_{1,1/2}$ or $\czkMIP^*_{1,1/2}$) protocol for a language $L$,
    the question-reduced protocol $P^{\QR} = (Q^{\QR}, V^{\QR})$ obtained by applying 
    the question-reduction technique of \cite{natarajan2023quantum},
    is a $\pzkMIP^*_{1,1/2}[q', a', t_Q',t_V']$  (resp. $\szkMIP^*_{1,1/2}$ or $\czkMIP^*_{1,1/2}$) protocol
    for $L$, where $q', a', t_Q'$ and $t_V'$ can be derived from \cref{prop:qr_summary}.
\end{proposition}

\begin{proof}

    Since $q(n)$ and $a(n)$ are upper-bounded by the sampling time and verification time, respectively,
    we set $q(n) = a(n) = O(\poly(n))$, 
    which implies that $q'(n) = O(\polylog(n))$, $t_Q'(n) = O(\polylog(n))$, and $t_V'(n) = O(\poly(n))$.
Given a $\pzkMIP^*$ protocol $P$ for the instance $\ipt$ with $n = \abs{\ipt}$.
    The question-reduced protocol $P^{\QR}$ executes the Introspection game detailed in \cite[Section 8]{re} with the two provers. 
    The honest behaviour of the provers are given in \cite[Section 8.3.2]{re}.
    In this game, all the questions are about Pauli measurements on $\ket{\EPR}^{\x q(n)} = \ket{\EPR}^{\x \poly(n)}$ except for questions $\mathsf{INTRO}$, $\mathsf{READ}$, and $\mathsf{SAMPLE}$. When a prover gets these three questions, they need to measure $\ket{\EPR}^{\x \poly(n)}$ in the Z basis, compute the questions $x,y$ of $P$ based on the measurement results, and answer the questions $x,y$ by measuring an auxiliary quantum system. 
    
    Let $S$ be the simulator for $P$ and $S'$ be the simulator for $P^{\QR}$. 
    If $S'$ doesn't receive questions involving $\mathsf{INTRO}$, $\mathsf{READ}$, or $\mathsf{SAMPLE}$,
    $S'$ can simulate the honest Pauli measurements on $\ket{\EPR}^{\x \poly(n)}$
    in time $\poly(n)$. 
    On the other hand, if $S'$ gets questions involving $\mathsf{INTRO}$, $\mathsf{READ}$, or $\mathsf{SAMPLE}$, it 
    simulates the Pauli measurements to compute the questions $x,y$ of $P$. This computation involves computing a conditional linear function (See \cite[Sections 4 \& 8]{re} for details) on the measurement result.
    Since this conditional linear function can be evaluated by $V$ too, 
    it takes time $\poly(n)$.
    Then $S'$ calls $S$ on $x,y$ to get the answers for $V^{\QR}$. $S'$ gives perfect simulation because $S$ is perfect.

    Suppose the output distribution of $S$ is statistically close to the perfect distribution. In that case, the output distribution of $S'$ is even closer to the perfect correlation because the special questions $\mathsf{INTRO}$, $\mathsf{READ}$, and $\mathsf{SAMPLE}$ are only sampled with a constant probability, and all the other questions can be answered perfectly. This proves the statistical zero-knowledge property.
    Suppose a distinguisher can distinguish the output distribution of $S'$ from the perfect correlation. In that case, it must be able to distinguish the conditional distribution of the special questions, because the conditional distributions of other questions are the same. Then this distinguisher can be used to distinguish the output distribution of $S$ from the perfect one. This proves the computational zero-knowledge property. 

    Let the oracularized protocol of $P^{\QR}$ be $P^{\OR}$. All the questions pairs of $P^{\OR}$ only involve measuring EPR pairs except for the special question pairs $(\mathsf{INTRO_A}, \mathsf{INTRO_B})$, $(\mathsf{INTRO_{A/B}}, \mathsf{READ_{A/B}})$, 
$(\mathsf{INTRO_{A/B}}, \mathsf{SAMPLE_{A/B}})$
    and $(\mathsf{SAMPLE_{A/B}}, \mathsf{PAULI_Z})$.
    Next we construct a simulator $S''$ for $P^{\OR}$. If at most one question that $S''$ gets is a special pair, then
    $S''$ can just work as $S'$ above.
    If both question pairs are special, since these questions only involve Pauli-Z measurements on EPR pairs, 
    the measurement outcomes for both question pairs are consistent, either $(x,x)$, $((x,y),x)$, $((x,y),y)$ or $((x,y),(x,y))$ for the protocol $P$. Then $S$ can be used to answer these questions. The arguments for the zero-knowledge properties of $P^{\OR}$ is similar as those for $P^{\QR}$. Finally, 
    the zero-knowledge properties of parallel repeated $P^{\OR}$ follows from \cref{prop:parallel_repeat}, which completes the proof.
\end{proof}

\section{Answer reduction preserves ZK}
\label{sec:ar}
In this section, we prove that the tightened low-degree-code-based answer reduction technique summarized in \cite[Theorem 6.2]{dong2023computational} preserves zero knowledge against honest verifiers. 
Later, we apply this answer reduction technique iteratively to reduce the answer length of a $\pzkMIP^*$ protocol from polynomial to a constant.
We also show the Hadamard-code-based answer reduction summarized in \cite[Theorem 6.8]{dong2023computational} also preserves zero knowledge, but it cannot be used for succinct zero-knowledge proof.
The tightened answer reduction in \cite[Theorem 6.2]{dong2023computational} is improved upon \cite[Theorem 52]{natarajan2023quantum} as the tightened version can be applied recursively to reduce the answer size below $\log(n)$.

We first introduce Probabilistically Checkable Proof system (PCP).
\begin{definition}[PCP verifier, Definition 2.1.1 of \cite{harsha2004robust}]
    A verifier is a probabilistic polynomial-time algorithm $V$ that, on $\ipt$ of length $n$, tosses
 $r =r(n)$ random coins $R$ and generates a sequence of $q = q(n)$ queries $I = (i_1,...,i_q)$ and a circuit
 $D:\{0,1\}^q \to \{0,1\}$ of size at most $d(n)$.
 
 We think of $V$ as representing a probabilistic oracle machine that queries its oracle $\pi$ for the positions
 in $I$, receives the $q$ answer bits $\pi|_I = (\pi_{i_1}, \ldots, \pi_{i_q})$, and accepts iff $D(\pi|_I) = 1$. We write $(I,D) \overset{R}{\leftarrow} V(\ipt)$ to denote the queries and circuit generated by $V$ on $\ipt$ with random coin tosses.
 We call $r$ the randomness complexity, $q$ the query complexity, and $d$ the decision complexity of
 $V$.
\end{definition}

\begin{definition}[PCP system, Definition 2.1.2 of \cite{harsha2004robust}]
    For a function $s:\mathbb{Z}^+\rightarrow[0,1]$,
    a verifier V is a probabilistically
 checkable proof system for a language L with soundness error $s$ if the following two conditions hold for
 every string $\ipt$.

    \begin{itemize}
        \item Completeness: If $\ipt \in L$, then there exists a proof $\pi$ such that the verifier $V$ accepts the proof $\pi$ with probability $1$. That is,
        \begin{equation}
            \exists\pi,\Pr_{(I,D)
            \overset{R}{\leftarrow} V(\ipt)}[D(\pi|_I)=1]=1,
        \end{equation}
        \item Soundness: If $\ipt \notin L$, then for every proof oracle $\pi$, the verifier $V$ accepts with probability strictly less than $s(|\ipt|)$. That is,
        \begin{equation}
            \forall\pi,\Pr_{(I,D)\overset{R}{\leftarrow} V(\ipt)}[D(\pi|_I)=1]<s(|\ipt|).
        \end{equation}
    \end{itemize}
\end{definition}

Suppose $P(\ipt)$ is an instance of an $\MIP^*$ protocol, in which $Q(\ipt)$ samples questions $x$ and $y$ and expects answers $a$ and $b$ from the provers. 
The goal of answer reduction is to reduce the answer length $a(n)$ exponentially.
The key observation of the tightened answer reduction technique is that $V$ has two phases. In the first phase,
$V$ computes a predicate circuit $C^{\ipt}_{x,y}$, and in the second phase, it outputs the output of the predicate circuit 
on answers $a,b$.
Hence, instead of sending back answers $a,b$, the provers compute encodings of $a,b$ and a proof that $a,b$ satisfy $C^{\ipt}_{x,y}$ for $V$ to query, which reduces the answer length exponentially.

Specifically, the answer reduction procedure first oracularizes $P$ to get $P^{\OR}$ such that one prover gets questions $(x,y)$ with a constant probability, and this prover needs to send back $(a,b)$ such that $V(x, y, a, b) = 1$.
Then, the procedure modifies $S^{\OR}$ to get a new sampler $S^{\AR'}$ that also samples queries to the encodings of $a,b$ and the proof $\pi$ of $V(x,y,a,b) = 1$ as a question to a prover. The other prover will get $x$ or $y$ with queries to the encoding of their answer.
The new decider $D^{\AR'}$ runs the PCP checks on the first prover's answers and checks the consistency between the prover's answers. Then the procedure applies oracularization and parallel repetition to $P^{\AR'} = (S^{\AR'}, D^{\AR'})$ to improve the soundness to $<1/2$. 
The final protocol is denoted by $P^{\AR} = (S^{\AR}, D^{\AR})$.

In the complete case, one prover computes $(a,b)$ for question $(x,y)$, and the other prover computes the answer $a$ or $b$
for the question $x$ or $y$, respectively.
Then, they encode their answers using the low-degree code.
The prover, who computes $(a,b)$, further computes a PCP proof $w$ for $C^{\ipt}_{x,y}(a,b) = 1$, and encodes the proof $w$ with the low-degree code.
The proof $w$ is the values of all the internal wires of the circuit $C^{\ipt}_{x,y}$.
In addition to sending questions $x,y$, the verifier also queries the encodings of $a$, $b$ and $w$.
The provers simply answer those queries honestly.
Finally, $P^{\AR}$ is parallely repeated to reduce soundness to at most $1/2$.

In summary, by \cite[Theorem 6.2]{dong2023computational},  we have
\begin{proposition}
    \label{prop:tight_ar_summary}
    If $P = (Q,V)$ is a $\MIP^*_{1, 1/2}[q,a,t_Q, t_V, d_V]$ protocol for Language $L$,
    the answer-reduced protocol $P^{\AR} = (Q^{\AR},V^{\AR})$ is also a $\MIP^*_{1, 1/2}$ protocol for $L$
    with
\begin{description}
    \item[PCC property:] If $\ipt \in L$ and there is a perfect PCC strategy for $P(\ipt)$, then there is a perfect PCC strategy for $P(\ipt)^{AR}$, 
    \item[Question length:] $q'(n) = \polylog(d_V(n))(2q(n)+\polylog(d_V(n)))$,
    \item[Answer length:] $a'(n) = \polylog(d_V(n))$,
    \item[Sampling time:] $t_Q'(n) = \polylog(d_V(n))(t_Q(n) + \polylog(d_V(n)))$,
    \item[Verification time:] $t_V'(n) = \polylog(d_V(n))(t_Q(n)+t_V(n)+\polylog(d_V(n)))$, and
    \item[Decision complexity:] $d_V'(n) = \polylog(d_V(n))$.
\end{description}
\end{proposition}

Below, we establish the zero-knowledge property of answer reduction.
\begin{proposition}
    \label{prop:ans_redu}
    If $P = (Q,V)$ is a $\pzkMIP^*_{1, 1/2}[q,a, \poly, \poly, d_V]$ (resp. $\szkMIP^*_{1,1/2}$ or $\czkMIP^*_{1,1/2}$) protocol for a language $L$,
    then the answer-reduced protocol $P^{AR} = (Q^{AR},V^{AR})$ obtained by applying the tighter answer-reduction technique in \cite[Theorem 6.2]{dong2023computational} 
    is a $\pzkMIP^*_{1, 1/2}[q', a', t_Q' ,t_V', d_V']$ (resp. $\szkMIP^*_{1,1/2}$ or $\czkMIP^*_{1,1/2}$) protocol for $L$ against honest verifiers, where the expressions of $q', a', t_Q', t_V', d_V'$ follows \cref{prop:tight_ar_summary}. 
\end{proposition}


\begin{proof}
    The fact that $P$ is $\pzkMIP^*$ implies that there is a simulator $S$ such that for every 
    instance $\ipt$ and questions $x,y$ sampled by $Q(\ipt)$, the simulator can sample answers $a,b$ from a perfect quantum correlation $p_{\ipt}$ satisfying $V(x,y,a,b) = 1$ in time $\poly(n)$ where $n = \abs{\ipt}$.

    Consider the answer-reduced protocol $P^{\AR}$. We will construct a simulator $S'$ for am honest verifier of $P^{\AR}$.
    The honest verifier samples the questions $x,y$ from the questions of $P^{\OR}$ and other queries of the encodings. $S'$ parses the questions from $P$ and calls $S$ on those questions to get answers $a,b$. This step takes polynomial time because the questions from $P^{\OR}$ has a certain form.
    Then $S'$ computes the circuit $C^{\ipt}_{x,y}(a,b) = 1$, and computes the proof $w$ for $C^{\ipt}_{x,y}(a,b) = 1$. As we have discussed above, $C^{\ipt}_{x,y}$ can be computed by $V$ in $\poly(n)$ time, and 
    $w$ can be computed by tracing the evaluation of $C^{\ipt}_{x,y}(a,b)$, so the running time of $S'$ is $\poly(n)$,
    and the size of $w$ is $\poly(n)$. 
    Moreover, the additional queries from $V^{\AR}$ about the low-degree encodings of $a,b$ and $w$ can be answered in $\poly(n)$ time.
    Therefore, $S'$ can sample from a new perfect correlation $p_{\ipt}'$ for $V^{\AR}$ in $\poly(n)$ time.


    For statistical zero-knowledge, if $p_{\ipt}$ is statistically close to the perfect correlation of $V$, then $p'_{\ipt}$ is also statistically close to the perfect correlation of $V^{\AR}$ because even if $S$ generates wrong answers, there is a nonzero probability that $S'$ can answer $V^{\AR}$'s queries correctly.
    
    Lastly, if $p'_{\ipt}$ is computationally distinguishable from the perfect correlation, we can construct a new distinguisher for $p_{\ipt}$ by running the computation of $S'$ on $a,b$ no matter $a,b$ are correct or not.
    and feed the output to the distinguisher for $p'_{\ipt}$. This establishes the computational zero-knowledge property.
    
\end{proof}

The Hadamard-code-based answer reduction technique summarized in \cite[Theorem 6.8]{dong2023computational} uses Mie's PCPP \cite{mie2009short}.
We give an overview of Mie's PCPP and this answer reduction technique in \cref{appd:ans_red}.
In summary, by \cite[Theorem 6.8]{dong2023computational}, we have
\begin{proposition}
    \label{prop:had_ar_summary}
    If $P = (Q,V)$ is a $\MIP^*_{1, 1/2}[q,a,t_Q, t_V, d_V]$ protocol for Language $L$,
the answer-reduced protocol $P^{\AR} = (Q^{\AR},V^{\AR})$ is also a $\MIP^*_{1, s}$ protocol for $L$ with
\begin{description}
    \item[Soundness:] a constant $s$,
    \item[Question length:] $q'(n) = O(q(n) + a(n) + (t_V(n) + 2^{a(n)})\polylog(t_V(n)+2^{a(n)}))$,
    \item[Answer length and decision complexity:] $a'(n) = d_V'(n) = O(1)$,
    \item[Sampling time:] $t_Q'(n) = O(t_Q(n) + a(n) + (t_V(n) + 2^{a(n)})\polylog(t_V(n)+2^{a(n)}))$, and
    \item[Verification time:] $t_V'(n) = \poly(n+ q(n), a(n), \log(t_V(n) + 2^{a(n)}))$.
\end{description}
\end{proposition}

\begin{proposition}
    \label{prop:ans_red_had}
    Let $P = (Q,V)$ is a $\pzkMIP_{1,1/2}^*[q, a, \poly, \poly, d_V]$ (resp. $\szkMIP^*_{1,1/2}$ or $\czkMIP^*_{1,1/2}$) protocol for a language $L$, and let $P^{\AR} = (Q^{\AR},V^{\AR})$
    be the $\MIP^*_{1,s}[q', O(1), t_Q', t_V', O(1)]$ protocol obtained by applying the 
    Hadamard-code-based
    answer-reduction technique of \cite{dong2023computational}, where $s$ is a constant and expressions of $q', t_Q', t_V'$ can be derived from \cref{prop:had_ar_summary}.
    Then if $q'(n) = \poly(n)$, $P^{AR}$
    is a $\pzkMIP^*_{1,s}$ (resp. $\szkMIP^*_{1,s}$ or $\czkMIP^*_{1,s}$) protocol for $L$ against \emph{honest} verifiers.
\end{proposition}
The proof of Proposition \cref{prop:ans_red_had} is similar to that of Proposition \cref{prop:ans_redu}, and we provide it in \cref{appd:ans_red}.
From the statement, we can see the limitation of this answer reduction technique is that it can only be applied to protocols with $O(\log n)$ answer size and $O(\poly(n))$ verification time because otherwise the question size and question sampling time will be superpolynomial.
Moreover, we cannot apply it to get succinct $\pzkMIP^*$ protocol, because when $t_V(n) = \poly(n)$,
the new question length is dominated by $t_V = \poly(n)$, i.e.
$q'(n) = \poly(n)$, even if $a(n) = O(\log n)$.

\section{Parametrized PZK transformation}
\label{sec:pzk}
In \cite{mastel2024two}, Mastel and Slofstra show that we can transform a BCS$-\MIP^*$ protocol into a two-prover perfect zero-knowledge $\MIP^*$ protocol using a version of the perfect zero-knowledge protocol for 3SAT due to \cite{dwork1992low}, modified to prove soundness against quantum provers. To prove our result, we need to repeatedly oracularize the protocol as part of the answer reduction step. The perfect zero-knowledge property is not preserved under oracularization. A cheating verifier can use the oracle questions to simultaneously learn the players' answers to two question pairs. The verifier may not be able to simulate the resulting correlation. To avoid this problem, we can further modify the \cite{dwork1992low} protocol to make it robust against oracularization. We begin by applying a transformation called
oblivation. 
\begin{definition}[Obliviation] \label{def:obliviation}
    Given a BCS $B$ with variables $X$ and constraints $\{(V_i,C_i)\}^m_{i=1}$ and let $k \geq 1$, let $Z =
    X\times [k]$, and $U_i = V_i\times [n]$ for any $1\leq
    i\leq m$. We denote
    $(x,i)$ by $x(i)$. Let $E_i \subseteq \{\pm1\}^{U_i}$ be the set of
    assignments $\phi$ to $U_i$ such that the assignment $\psi$ to $V_i$
    defined by $\psi(x) = \psi(x(1))\cdots \psi(x(n))$ is in $C_i$. The
    obliviation of $B$ of degree $k$ is the constraint system, denoted
    $\Obl_k(B)$ is the BCS with variables $Z$ and constraints $\{(U_i,E_i)\}_{i=1}^m$. We call the variables of $\Obl_k(B)$ oblivious variables.
\end{definition}

Obliviation has the following properties:
\begin{lemma}[\cite{mastel2024two} Lemma 9.4]\label{lem:obliviate}
    Suppose $B$ is a BCS, and let $\Obl_k(B)$ be the obliviation of $B$ for some $k \geq 1$, using the notation from \Cref{def:obliviation}. Then:
    \begin{description}
        \item[Completeness] If the constraint-constraint game for $B$ has a perfect finite-dimensional quantum strategy, then the constraint-constraint game for $\Obl_k(B)$ has a perfect finite dimensional quantum strategy.

        \item[Soundness] If the constraint-constraint game for $\Obl_k(B)$ has a quantum strategy with a winning probability of $1-\epsilon$, then the constraint-constraint game for $B$ has a quantum strategy with a winning probability of $1-\epsilon$.
    
        \item[Obliviousness] If there is a perfect finite dimensional quantum strategy for the constraint-constraint game for $B$, then there is a perfect strategy for the constraint-constraint game for $\Obl_k(B)$ such that, for any set of less than $k$ variables, the players' answers are uniformly random. 
        \end{description}
\end{lemma}

A \textbf{permutation branching
program} of width $5$ and depth $d$ on a set of variables $X$  is a tuple $P =
(X, \{(x_i,\pi_{1}^{(i)},\pi_{-1}^{(i)})\}_{i=1}^{d}, \sigma)$ where 
$x_i \in X$ and $\pi_1^{(i)}, \pi_{-1}^{(i)}$ are elements of the
permutation group $S_5$ for all $1 \leq i \leq d$, and $\sigma \in S_5$ is a
5-cycle. A permutation branching program $P$ defines a map $P : \Z_2^{X} \to
S_5$ via $P(\phi) = \prod_{i=1}^d \pi^{(i)}_{\phi(x_i)}$. A program $P$
\textbf{recognizes a constraint $C \subseteq \Z_2^{X}$} if $P(\phi) = \sigma$
for all $\phi \in C$, and $P(\phi) = e$ for all $\phi \not\in C$, where $e$ is
the identity in $S_5$.
\begin{theorem}[Barrington \cite{Barrington86}]
    Suppose a constraint $C \subseteq \Z_2^X$ is recognized by a depth $d$ fan-in 2
    boolean circuit. Then $C$ is recognized by a permutation branching program
    of width $5$ and depth $4^d$ on the variables $X$. 
\end{theorem}
Assume for the rest of the chapter that we have a canonical way of turning constraints described by fan-in 2 boolean circuits into permutation branching
programs using Barrington's theorem.

The final ingredient is randomizing tableaux, which are described using
constraints of the form $x_1 \cdots x_n = \gamma$, where the variables
$x_1,\ldots,x_n$ take values in $S_5$, $\gamma$ is a constant in $S_5$, and the
product is the group multiplication. We can now
define the $\ell$-row randomizing tableaux, modifying the definitions from \cite{dwork1992low} and \cite{mastel2024two}.
\begin{definition}\label{def:Tab}
    Let $B$ be a BCS with variables $X$ and constraints $\{(V_i,C_i)\}_{i=1}^m$ be a BCS, where each $C_i$ is described
    by a fan-in 2 boolean circuit. Let $\ell \geq 4$ and $P_i = (V_i,\{(x_{ij},\pi_{1}^{(ij)},
    \pi_{-1}^{(ij)})\}_{j=1}^{d_i},\sigma_i)$ be the permutation branching program
    recognizing $C_i$. For each $i \in [m]$, let
	\begin{equation*}
		W_i = V_{i}\sqcup\{T_{i}(p,q) : (p,q) \in [\ell] \times [d_i]\} \sqcup \{r_{i}(j,k): (j,k) \in [\ell-1] \times [d_i-1]\},
    \end{equation*}
    where $T_i(p,q)$ and $r_i(j,k)$ are new permutation-valued variables (and thus
    represent 7 boolean variables each), and let
	\begin{equation*}
        Y = X\sqcup\{T_{i}(p,q),r_{i}(j,k):(i,p,q,j,k)\in{[m]\times[\ell]\times[d_i]\times[\ell-1]\times[d_i-1]}\}
    \end{equation*}
    be the union of all the original and new variables. The variables $T_i(p,q)$
    are called tableau elements, and the variables $r_i(j,k)$ are called randomizers.

    Let $D_i$ be the constraint on variables $W_i$ which is the conjunction of the
    following clauses:
	\begin{enumerate}
		\item $T_{i}(1,q) = \pi^{(iq)}_{x_q}$ for all $q \in [d_i]$,
		\item $T_i(p+1,q) = r_i(p,q-1)^{-1}T_i(p,q)r_i(p,q)$ for $q \in [d_i]$ and $p \in [\ell-1]$, where
            we use the notation $r_i(p,0) = r_i(p,d_i) = e$, 
		\item $\prod_{1\leq q\leq d_i}T_i(\ell,q) = \sigma_i$, and
        \item a trivial constraint (meaning that all assignment are allowed) on any pair $x,y$
            of original or permutation-valued variables which do not appear in one of the above
            constraints.
	\end{enumerate}
    We get a constraint system with boolean variables representing the permutation valued variables, and constraints $\{(W_{i},D_{i})\}_{i=1}^{m}$. We further let $\{(W_{ij},D_{ij})\}_{j=1}^{m_i}$
    be a list of the clauses in (1)-(4) making up $D_i$. The \textbf{tableau} of $B$ is $\Tab_{\ell}(B)$ with variables $Y$ and constraints $\{(W_{ij},D_{ij})\}_{i\in [m], j\in
    [m_i]}$. For a fixed $i$ and $p$, we call the set of elements $\{T_i(p,q): 1\leq q\leq d_i\}$ a \textbf{row} of the tableau for constraint $i$.
\end{definition}

Our definition allows for $\ell$-row tableaux for any $\ell\geq 4$, while the tableaux in \cite{mastel2024two} are always $4$-row. The additional rows will be key in preserving the perfect zero knowledge property after multiple rounds of oracularization.

Composing the tableau construction with the obliviation transformation, we get the parametrized PZK transformation. The properties of this transformation are collected in the following lemma.
\begin{lemma}\label{lem:tab}
    Suppose $B$ is a BCS with constraints of size $V$, let $\ell\geq 4$ and $k\geq 5$, and let $\Tab_{\ell}(\Obl_{k}(B))$ be the tableau of the obliviation $\Obl_{k}(B)$, using the notation from \Cref{def:Tab} and \Cref{lem:obliviate}. Then:
    \begin{description}
        \item[Completeness] If the constraint-constraint game $G_1$ for $B$ has a perfect finite-dimensional quantum strategy, then the constraint-constraint game $G_2$ for $\Tab_{\ell}(\Obl_{k}(B))$ has a perfect finite dimensional quantum strategy.

        \item[Soundness] If the constraint-constraint game $G_2$ for $\Tab_{\ell}(\Obl_{k}(B))$ has a quantum strategy with a winning probability of $1-\epsilon$, then the constraint-constraint game $G_1$ for $B$ has a quantum strategy with winning probability $1-\poly(k,\ell,2^V)\epsilon$.
    
        \item[Perfect Correlation] There is a perfect correlation $p$ for the constraint-constraint game $G_2$ corresponding to $\Tab_{\ell}(\Obl_k(B))$ such that $p$ is a quantum correlation if and only if $G_1$ has a perfect quantum strategy. Moreover, if $\ell \geq 8$ and $k\geq 9$, then there is a perfect correlation $p^{OR}$ for the oracularization $G_2^{OR}$ of $G_2$ such that $p^{OR}$ is a quantum correlation if and only if $G_1$ has a perfect quantum oracularizable strategy.
        \end{description}
\end{lemma}


\begin{proof}
    The proof of completeness is the same as the proof of completeness of the obliviated tableau transformation in \cite{mastel2024two}. The soundness follows by noting that $\Tab_{\ell}(\Obl_k(B))$ has rows of length $2^V$ in the worst case and is constructed from $B$ by applying a classical homomorphism and subdivision transformation from \cite{mastel2024two}. 

    Construct the correlation $p$ for $G_2$ in the same way as the correlation in Proposition 9.12 of \cite{mastel2024two}, but with the potentially larger tableau and higher number of oblivious variables. In the proof of Proposition 9.12, the authors show that any non-scalar element of the subspace of the players' answers in the tableau game either contain a randomizer or have degree at most $4$ in the oblivious variables. This property still holds, as the tableau we construct is at least the same size as the \cite{mastel2024two} tableau and has at least the same degree of obliviation. By the same argument as in Proposition 9.12 of \cite{mastel2024two}, $p$ is quantum if and only if $G_1$ has a perfect quantum strategy.
    
    When $\ell \geq 8$ and $k\geq 9$, due the extra rows in the tableau, the non-scalar elements of the subspace of the players' answers in $G'$ either contain $5$ randomizers or have degree at most $4$ in the oblivious variables. Because we assume that any pair of questions may be asked simultaneously, oracularization of the game amounts to doubling the number of questions that a cheating verifier may ask each player. A cheating verifier may use this to eliminate $4$ randomizers from an element of the subspace of the players' answers by requesting two pairs of randomizers from one player. Similarly, a cheating verifier may gain access to elements of degree $8$ in the oblivious variables by asking for two question pairs that have degree $4$. Thus, the elements of the subspace of the players answers in $G_2^{OR}$ either contain a randomizer or have degree at most $8$ in the oblivious variables. Define the correlation $p^{OR}$ by the extending the procedure in Proposition 9.12 of \cite{mastel2024two} to oracle questions in the natural way. By a similar argument to that in Proposition 9.12 of \cite{mastel2024two}, $p^{OR}$ is quantum if and only if $G_1$ has a perfect oracularizable quantum strategy. 
\end{proof}
Note that if we oracularize the game $G_2^{OR}$ again, the only valid questions for each player will be oracle questions paired with type $A$ or $B$ questions. Sampling these answers is the same as sampling the answer to an oracle question in $G_2^{OR}$. Thus, adding more rounds of oracularization preserves the existence of the correlation in \Cref{lem:tab}.

\section{Putting everything together}
\label{sec:all}
In \cite{mastel2024two}, Mastel and Slofstra's construction starts with the $\MIP^*$ protocol for $\RE$ 
with $O(1)$-length questions, $\polylog$-length answers, perfect completeness and soundness $< 1/2$ 
from \cite{natarajan2023quantum}.
Then they apply their modified 3-step transformation initially proposed in \cite{dwork1992low} to turn this $\MIP^*$ protocol into a $\pzkMIP^*$ protocol
with $\polylog$-length questions, $O(1)$-length answers, perfect completeness, and soundness $\leq 1 - \poly^{-1}$.
In the end, parallel repetition is applied to reduce soundness to a constant while preserving the perfect zero-knowledge property, but the question length and answer length are increased to $\poly$.
Our proof of \cref{thm:main} can be viewed as an extension of that proof.
\begin{proof}[Proof of \cref{thm:main}]
    We first show that there exists a Turing machine $M$ that computes a $\pzkMIP^*[\polylog, O(1)]$
    protocol for $\RE$ with perfect completeness and soundness $< 1/2$.

        \begin{algorithm}[H]
        \caption{The Turing machine $M$}
        \begin{algorithmic}
            \State $P_{\text{base}} = \text{ the } \text{BCS-}\MIP^*_{1,1/2}[O(1), \polylog, O(1), \poly(n)] \text{ protocol for } \RE$ \cite{natarajan2023quantum};  
            \State $P_{3\text{SAT}} = \text{ the } \text{BCS-}\MIP^*_{1,1-1/\poly(n)}[\poly(n),O(1),\poly(n)] \text{ protocol for } \RE$ \cite{mastel2024two}; 
            \State $P_{ZK} = \Tab_{8}(\Obl_{9}(P_{3\text{SAT}}))$;
            \State $P_{ZK}^{cv} = \text{CC-to-CV-transform}(P_{ZK})$; \Comment{Constraint-Constraint game to Constraint-variable game transformation.}
            \State $P^{\PR} = (P_{ZK}^{cv})^{\otimes k}$; \Comment{$k$-fold parallel repeat of $P_1^{cv}$ to make its soundness $<1/2$.}
            \State $P_0 = \text{Question-reduce}(P^{\PR})$;
            \State $i = 0$;
            \While{Answer size of $P_i > \kappa_0$} \Comment{$\kappa_0$ is a constant determined later.}
                \State $P_{i+1} = \text{Answer-reduce}(P_i)$;
                \State $i = i+1$;
            \EndWhile
            \State Output $P_{i+1}$;
        \end{algorithmic}
    \end{algorithm}

    In \cite{mastel2024two}, the authors show by reduction form the protocol for $P_{\text{base}}$ for RE in \cite{natarajan2023quantum} that there is a BCS protocol $P_{3\text{SAT}}$ for the halting problem with polynomial length questions, where every constraint is a conjunction of three literals and the gap is inverse polynomial. By \cite[Theorem 9.15]{mastel2024two}, the tableau $P_{ZK}$ is a two-prover one-round nonlocal $\pzkMIP^*_{1,1-1/\poly}[\poly$, $O(1)$, $\poly$, $O(1)]$ protocol for $L \in \RE$.
    
    Since each instance $P_{ZK}(\ipt)$ is a constraint-constraint BCS game,  by \cref{thm:cc_to_cv}, the new protocol $P_{ZK}^{cv}$ is a constraint-variable BCS protocol.
    By \cref{thm:cc_to_cv}, the question length and the answer length of $P_{ZK}^{cv}$ remain $\poly(n)$ and $O(1)$ respectively. 
    Consequently, the verification time and the sampling time also remain the same.
    If $\ipt \in L$, $\val^*(P_{ZK}^{cv}(\ipt)) = 1$.
    If $\ipt \notin L$, $\val^*(P_{ZK}^{cv}(\ipt)) \leq 1- 1/\poly(n)$ by noticing $C = O(1)$ in the soundness part of \cref{thm:cc_to_cv}.
    Moreover, if $\ipt \in L$, since a perfect finite-dimensional quantum correlation of $P_{ZK}$ can be simulated in polynomial time by a simulator, so can 
    a perfect finite-dimensional quantum correlation of $P_{ZK}^{cv}(\ipt)$.
    Therefore, $P_{ZK}^{cv}$ is a $\pzkMIP^*_{1, 1-1/\poly}[\poly, O(1),\poly,O(1)]$ protocol.
    
    Then, by choosing $k = \poly(n)$, $P^{\PR}$ is a $\pzkMIP^*_{1,1/2}[\poly$, $\poly, \poly, \poly]$ protocol for $L$. The soundness $1/2$ follows from \cref{thm:polyparallel} and perfect zero-knowledge follows from \cref{prop:parallel_repeat}.

    After applying question reduction to $P^{\PR}$, $P_0$ is a $\pzkMIP^*_{1,1/2}[\polylog$, $\poly, \polylog, \poly]$ protocol for $L$. The parameters of the protocol and the soundness $1/2$ follows from \cref{prop:qr_summary} and perfect zero-knowledge follows from \cref{prop:que_redu}.


    Now, we are ready to give the description of a $\pzkMIP^*_{1,1/2}[\polylog, O(1), \poly, \poly]$ protocol $\bar{P}=(\bar{Q},\bar{V})$ for $L$.
    The sampler $\bar{Q}$ and decider $\bar{V}$ both start by running $M$.
    Then $\bar{Q}$ runs the sampler of the output protocol of $M$,
    and $\bar{V}$ runs the sampler of the output protocol of $M$.


    The parameters of $P_0$ are $q_0(n) = \polylog(n)$, $a_0(n) = \poly(n)$, $t_{Q,0}(n) = \polylog(n)$
    and $t_{V,0}(n) = d_{V,0}(n) = \poly(n)$.
    Following the same runtime argument in the proof of \cite[Theorem 54]{natarajan2023quantum}, we can get that by picking $\kappa_0$ large enough the number of steps of the loop of $M$ is $m = O(\log\log\log(a_0(n)))$.
    Besides the $O(1)$ answer size and decision complexity, the question size, sampling time, and verification time of $\bar{P}$ are:

    \textbf{Question size.} The question size follows the recursive relation 
\begin{align*}
    q_{i+1}(n) =  \polylog(d_{V,i}(n))(2 q_i(n) +  \polylog(d_{V,i}(n))).
\end{align*}
Since $d_{V,{i+1}}(n) = \polylog(d_{V,i}(n))$, we can bound 
\begin{align*}
    q_{m}(n) &= \polylog(d_{V,{m-1}}(n))(2 q_{m-1}(n) + \polylog(d_{V,{m-1}}(n))) \\
    &= 2^m \cdot \prod_{i=0}^{m-1} [\polylog(d_{V,i}(n))] \cdot q_{0}(n) + \sum_{i=0}^{m-1} 2^{m-1-i} \prod_{j=i}^{m-1} [\polylog(d_{V,j}(n))] \polylog(d_{V,i}(n)) \\
    &\leq 2^m \polylog(n)q_{0}(n) + 2^m m \polylog(n) \polylog(d_{V,0}(n)) \\
    &= O(\polylog(n) \polylog(n) + \polylog(n)) = O(\polylog(n)),
\end{align*}
where we use
\begin{align*}
    &\prod_{i=0}^{m-1} [\polylog(d_{V,i}(n))] = \polylog(n) \polyloglog(n) \ldots O(1) \\
    &\leq \polylog(n) \cdot \polyloglog(n)^{\polylogloglog(n)} = \polylog(n).
\end{align*}

\textbf{Sampling time.} It takes $m$ iterations for the sampler to calculate $P_m$. 
By \cite[Theorem 6.2]{dong2023computational}, $\abs{Q_i} = \abs{Q_{i-1}}+O(1)$, $\abs{V_i} = \abs{V_{i-1}} + O(1)$, and
the $(i+1)$th iteration takes time $O( \abs{Q_i}+ \abs{V_i}) = O(\abs{Q_0} + \abs{V_0} +i)$.
    Hence, the total computation time of $M$ is $O( m(\abs{Q_{0}} + \abs{V_{0}}) + m^2)$.
The running time of ${Q_m}$ follows the relation
\begin{align*}
    t_{Q,m}(n) &= \polylog(d_{V,m-1}(n)) \cdot (t_{Q,{m-1}}(n) + \polylog(d_{V,{m-1}}(n))) \\
    &= \ldots \\
    &= \prod_{i=0}^{m-1} [\polylog(d_{V,i}(n))] \cdot t_{Q,0}(n) + \sum_{i=0}^{m-1} \prod_{j=i}^{m-1} [\polylog(d_{V,j}(n))] \cdot 
    \polylog(d_{V,{i}}(n)) \\
    &\leq \polylog(n) t_{Q,0}(n) + m \polylog(n) \polylog(d_{V,0}(n)) = \poly(n).
\end{align*}
Hence, the total sampling time of $\bar{Q}$ is $O( m(\abs{Q_0} + \abs{V_0}) + m^2 + \poly(n)) = \poly(n)$
as $\abs{Q_0}$ and $\abs{V_0}$ are constants.

\textbf{Verification time.} Similar to the previous case, the time to calculate $V_m$ is
$O( m(\abs{Q_0} + \abs{V_0}) + m^2)$.
The verification time of $V_m$ is
\begin{align*}
    t_{V,m}(n) &= \polylog(d_{V,{m-1}}(n))(t_{V,{m-1}}(n) + t_{Q,{m-1}}(n) + \polylog(d_{V,{m-1}}(n))) \\
    & = \ldots \\
    & = \prod_{i=0}^{m-1} [\polylog(d_{V,i}(n))] \cdot (t_{Q, 0}(n) + t_{V,0}(n)) + \sum_{i=0}^{m-1} \prod_{j=i}^{m-1} [\polylog(d_{V,j}(n))] \cdot 
    \polylog(d_{V,{i}}(n)) \\
    & \leq \polylog(n) (t_{V,0}(n) + t_{Q,0}(n)) + m \polylog(n) \polylog(d_{V,0}(n)) = \poly(n).
\end{align*}
Hence the total verification time is $O( m(\abs{Q_0} + \abs{V_0}) + m^2 + \poly(n)) = \poly(n)$.

Lastly, the protocol $\bar{P}$ has completeness 1 and soundness at most $1/2$ by \cref{prop:tight_ar_summary}.
The perfect zero-knowledge property starts from the observation that if the base protocol samples questions $i$ and $j$ for Alice and Bob, the oracularized question of the oracularized protocol is of the form $(i,j)$.
If the base protocol is oracularized $k$ times, the oracularized question is the of the form $((i,j), \beta_1,\beta_2, \ldots, \beta_{k-1})$ where each $\beta_\ell$ is a portion of $((i,j), \beta_1,\beta_2, \ldots, \beta_{\ell-1})$.
Hence, when a dishonest verifier executes the $k$-times oracularized protocol, the verifier can learn at most four answers of the base protocol: answers to $i,j$ from Alice and answers to $i',j'$ from Bob.
In our case, the base protocol $L_0$ is oracularized $m$ times. 
Therefore, we choose the degree of obliviation to be $9$ and the number of rows in the random tableaux to be $8$ so that the answers to the four base questions are random to the verifier.
Then the simulator of $\bar{P}$ can first randomly sample the answer to these base questions, which are from a quantum correlation by \cite[Theorem 9.15]{mastel2024two}, and the answers to the questions added from answer reduction can be calculated accordingly via classical postprocessing as in the proof of \cref{prop:ans_redu}.
\end{proof}
Note that the proof above also implies that $\RE = \MIP^*[\polylog, O(1)]$.
Moreover, if we follow the same argument as above, but do answer reduction before the loop and question reduction inside the loop, we get \cref{thm:am}.
Note that the communication complexity in \cref{thm:main,thm:am} are optimal up to a $\polylog$ factor by the $\log(n)$ communication lower bound in \cite[Theorem 64]{natarajan2023quantum}.
Also, the $\poly(n)$ verification time in \cref{thm:main,thm:am} are optimal.
This is because $\bar{V}$ needs to read in $\ipt$ as $\bar{Q}$ and $\bar{V}$ are independent.

\bibliographystyle{alpha}
\bibliography{ref}

\appendix

\section{Mie's PCPP and the proof of \Cref{prop:ans_red_had}}
\label{appd:ans_red}

The Hadamard-code-based answer reduction technique relies on the $\PCPP$ constructed in~\cite{mie2009short}. 
By inspecting the construction of Mie's PCPP explicitly, we find that the PCPP proof can be computed in polynomial time in the honest case, which allows for an efficient zero-knowledge simulation. 
We first review the basics of $\PCPP$ and the construction in \cite{mie2009short}, and then give an overview of the Hadamard-based answer reduction technique. Afterwards, we prove Proposition \cref{prop:ans_red_had}.

\paragraph{Overview of Mie's PCPP.}
A \emph{pair language} is a subset of $\set{0,1}^* \times \set{0,1}^*$ of the form $(x, y)$ where $x$ is called the explicit input and $y$ is called the implicit input. 
In the context of $\PCPP$, the verifier receives the explicit input $x$, but can only query $y$.
For a pair language $L$, we define $L(x)\equiv\{y:(x,y)\in L\}$. Based on these notions, we define $\PCPP$.

\begin{definition}[PCPP, Definition~2.2.1 in~\cite{harsha2004robust}]
    Let $s,\delta:\mathbb{Z}^+\rightarrow[0,1]$.
    A random verifier $V$ is a $\PCPP$ system for a pair language $L\subseteq\Sigma^*\times\Sigma^*$ with proximity parameter $\delta$ and soundness error $s$ if the following holds for every pair of strings $(x,y)$:
    \begin{itemize}
        \item Completeness: If $(x,y)\in L$, then there exists a proof $\pi$ such that the verifier $V$ accepts the oracle $y\circ\pi$ (i.e., the implicit input concatenated by the proof) with probability $1$. That is,
        \begin{equation}
            \exists\pi,\Pr_{(I,D)
            \overset{R}{\leftarrow} V(x,|y|)}[D((y\circ\pi)|_I)=1]=1,
        \end{equation}
        where $V$ receives the input $(x,|y|)$, after which $V$ outputs a circuit $D$ and a set of positions $I$ to query in the oracle $y\circ\pi$ using random coins $R$, and $D$ outputs a binary variable after probing $(y\circ\pi)|_I$, with $D=1$ representing an acceptance.
        \item Soundness: If $y$ is at least $\delta(|x|)$-far away from $L(x)$, then for every proof oracle $\pi$, the verifier $V$ accepts the oracle $y\circ\pi$ with probability strictly less than $s(|x|)$. That is,
        \begin{equation}
            \forall\pi,\Pr_{(I,D)\leftarrow V(x,|y|)}[D((y\circ\pi)|_I)=1]<s(|x|).
        \end{equation}
    \end{itemize}
\end{definition}

We denote the length of the explicit input as $|x|=n$, the length of the implicit input as $|y|=K$, and the total length as $|x|+|y|=m$, which are all written in binary values. Consider functions $r,q:\mathbb{Z}^+\rightarrow\mathbb{Z}^+$ and $t:\mathbb{Z}^+\times\mathbb{Z}^+\rightarrow\mathbb{Z}^+$.
In the PCPP, the randomness complexity, namely the number of coin tosses, is denoted as $r(m)$; the query complexity, namely the number of queried symbols of the oracle $y\circ\pi$, is denoted as $q(m)$; the time complexity for verification is denoted as $t(n,K)$.


In our construction of a succinct $\pzkMIP^*$ system, we shall utilize a PCPP verifier with a constant query complexity~\cite{mie2009short}. The parameters of this PCPP are given as follows.

\begin{theorem}[Mie's PCPP, Theorem~1 in~\cite{mie2009short}]
    \label{thm:mie_pcpp}
    Consider a pair language in $\mathsf{NTIME}(T)$ for a non-decreasing function $T:\mathbb{Z}^+\rightarrow\mathbb{Z}^+$. The length of any purported word $(x,y)$ of $L$ is denoted as $m=n+K$, where $n = \abs{x}$, and $K=\abs{y}$. Then, for any constants $s,\delta>0$, there exists a PCPP verifier with soundness error $s$ and proximity parameter $\delta$, such that
    \begin{itemize}
        \item randomness complexity $r(m)=\log{T(m)}+O(\log\log{T(m)})$;
        \item query complexity $q(m)=O(1)$;
        \item time complexity $t(n,K)=\poly(n,\log{K},\log{T(m)})$.
     \end{itemize}   
\end{theorem}

We shall explain that the proof in this PCPP system can be computed in polynomial time. To see this, we briefly explain its construction. Mie's PCPP starts from an efficient Ben-Sasson-type (BS-type) PCPP, which has a polylogarithmic verification time~\cite{harsha2004robust,ben2005short}. This PCPP has the following parameters.

\begin{theorem}[Theorem~1.3.2 in~\cite{harsha2004robust} and Theorem~2.5 in~\cite{ben2005short} (restated in Theorem~2 in~\cite{mie2009short})]
    Consider a pair language in $\mathsf{NTIME}(T)$ for a non-decreasing function $T:\mathbb{Z}^+\rightarrow\mathbb{Z}^+$. The length of any purported word $(x,y)$ of $L$ is denoted as $m=n+K$, where $n = \abs{x}$, and $K = \abs{y}$. Then, for any constants $s,\delta>0$, there exists a BS-type PCPP verifier with soundness error $s$ and proximity parameter $\delta$, such that
    \begin{itemize}
        \item randomness complexity $r(m)=\log{T(m)}+O(\log\log{T(m)})$;
        \item query complexity $q(m)=\polylog{T(m)}$;
        \item time complexity $t(n,K)=\poly(n,\log{K},\log{T(m)})$.
    \end{itemize}    
\end{theorem}

A convenient way to analyze a PCPP Verifier is through the corresponding  constraint graph~\cite{dinur2007pcp}.
\begin{definition}[Definition~1.5 in~\cite{dinur2007pcp}, Constraint graph]
    $\Graph=((V,E),\Sigma,\mathcal{C})$ is called a constraint graph if it satisfies the following conditions:
    \begin{enumerate}
        \item $(V,E)$ is an undirected graph;
        \item The vertex set $V$ also represents a set of variables over alphabet $\Sigma$;
        \item Each edge $e\in E$ carries a constraint $c(e)\subseteq\Sigma\times\Sigma$ over the variables corresponding to the linking vertices, and $\mathcal{C}=\{c(e)\}_e$.
    \end{enumerate}
    A constraint $c(e)$ is said to be satisfied by $(a,b)$ iff $(a,b)\in c(e)$. We call $(V,E)$ the underlying graph of $\Graph$.
\end{definition}

The actions of a PCPP verifier can be described by a constraint graph, as follows. For a pair language $L$ and a purported word $(x,y)$ of it, first set up a corresponding constraint graph $\Graph=((V,E),\Sigma,\mathcal{C})$, where the explicit input $x$ defines the underlying graph $(V,E)$ with constraint set $\mathcal{C}=L(x)$, $\Sigma$ is the alphabet of the variables, and $y \circ \pi$ is a purported assignment to the constraints $L(x)$. 
Then the verifier randomly samples a constraint in $\mathcal{C}$ and queries $
\sigma = y \circ \pi$ accordingly.
We denote $\mathsf{UNSAT}_\sigma(
\Graph)$ as the proportion of constraints it cannot satisfy. The constraint graph is satisfiable if there exists an assignment $\rho$ such that $\mathsf{UNSAT}_\rho(\Graph)=0$, which we denote as $\mathsf{UNSAT}(\Graph)=0$. Otherwise, we denote $\mathsf{UNSAT}(\Graph)=\min_{\sigma}\mathsf{UNSAT}_\sigma(\Graph)$.
The time complexity of the PCPP verification corresponds to the time the verifier samples an edge and checks the returned answers,
the query complexity corresponds to the number of variable values to read 
in one query, and the randomness complexity corresponds to amount of randomness to sample an edge.


To show how Mie constructs the final constraint graph, a final concept we need to use is the expander graph.
\begin{definition}[Expander graph]
    Given constants $\lambda,d$ satisfying $d\in\mathbb{N},\lambda<d$, we call a graph $\Graph$ an $(n,d,\lambda)$-graph if $\Graph$ is an $n$-vertex, $d$-regular graph, of which the second largest eigenvalue $\lambda(\Graph)<\lambda$. A family of graphs $\{\Graph_n\}$ is called an expander graph family if there exists constants $d\in\mathbb{N},\lambda<d$, such that for every $n$, $\Graph_n$ is an $(n,d,\lambda)$-graph.
\end{definition}





Mie's PCPP is constructed as follows.
 Suppose that verifier $V$ would like to verify that $\ipt \in L \in \NTIME(T)$, and suppose that
 there is a BS-type $\PCPP$ verifier $V^{BS}$ for $\ipt \in L$.
 The proof for the BS-type $\PCPP$ verifier $V^{BS}$ is constructed as follows. 
 Suppose $w$ is the witness for $\ipt \in L$. There must exist a circuit $C$ such that $C(\ipt, w) =1$.
 Then the proof is the evaluation table of $C$ on $(\ipt, w)$, which consists of the input and output values of all the gates in $C$, encoded in a low-degree polynomial.
 Mie's construction of PCPP is the following.
 \begin{enumerate}
    \item  Mie first oracularizes $V^{BS}$ to reduce the number of queries to $2$ at the cost of the increased alphabet size.
 In this step, the proof for the oracularized verifier can be computed in time $\poly(T)$.
    \item  Then Mie's construction applies his gap amplification procedure to the oracularized verifier, whose constraint graph is denoted by $\Graph$.
 Then the final constraint graph out of the gap amplification procedure preserves completeness of $\Graph$ and amplifies the soundness error $s$ to a constant value, say $1/2$.
    \begin{enumerate}
        \item In the first step, dummy questions are added to make the constraint graph expanding. The size of the proof is multiplied by a constant factor to include the answers of the dummy questions. The overhead for computing the proof is also a constant.
        \item The second step is graph powering, where multiple queries from the previous step are asked at the same time. The number of queries is a constant, so the size of the new proof is a polynomial of the size of the proof in the previous step, similarly for the computation time of the proof.
        \item The third step is to apply an assignment tester, which is similar to a PCPP, to the combined queries to reduce the alphabet size. 
In this step, PCPP proofs of the combined queries are amended to the proof from the previous step. 
The amendment is of a polynomial size and can be computed in polynomial time.
    \end{enumerate} 
The whole gap amplification procedure is repeated $\log\log(T)$ times, so the total blow-up of the size and computation time of the proof is $\poly(T)$.
    \item 
The last step is alphabet reduction where the constant-sized alphabet is reduced to the binary alphabet. 
The blow-up to the size and computation time of the proof is a constant.
 \end{enumerate}
 
\paragraph{Overview of the answer-reduction technique of \cite{dong2023computational}.}
Suppose $P(\ipt)$ is an instance of an $\MIP^*$ protocol, in which $Q(\ipt)$ samples questions $x$ and $y$ and expects answers $a$ and $b$ from the provers. 
The goal of Hadamard-code-based answer reduction is to reduce the answer length $a(\cdot)$ to a constant.
We first oracularize $P$ to get $P^{OR}$ such that one prover gets questions $(x,y)$ with a constant probability, and this prover needs to send back $(a,b)$ such that $V(x, y, a, b) = 1$.
Then, we apply answer reduction to $P^{OR}$ to get $P^{AR} = (Q^{AR}, V^{AR})$.
In the complete case, one prover computes $(a,b)$ for question $(x,y)$, and the other prover computes the answer $a$ or $b$
for the question $x$ or $y$, respectively.
Then, they encode their answers using the Hadamard code.
The prover, who computes $(a,b)$, further computes a Mie's $\PCPP$ proof $\pi$ for $V(x,y,a,b) = 1$, and encodes the proof $\pi$ with the Hadamard code.
In addition to sending questions $x,y$, the verifier also queries the encodings of $a$, $b$ and $\pi$.
The provers simply answer those queries honestly.
Note that one important property of Mie's PCPP that is not stressed in \cite{neexp,dong2023computational} is that the alphabet of applicable PCPPs must be binary.
If its alphabet is not a constant, we cannot reduce the answer length to a constant.
This means the PCPP developed in \cite[Section 10]{re} cannot be used for reducing the answer length to a constant because its alphabet is not a constant.

\begin{proof}[Proof of \cref{prop:ans_red_had}]
    We first focus on the case that $P$ is $\pzkMIP^*$. Then, there is a simulator $S$ such that for every 
    instance $\ipt$ and questions $x,y$ sampled by the verifier for $\ipt$, the simulator can sample answers $a,b$ from the perfect quantum correlation $p_{\ipt}$ satisfying $V(x,y,a,b) = 1$ in time $\poly(n)$ where $n = \abs{\ipt}$.

    Consider the answer-reduced verifier $P^{AR}$. We will use $S$ to construct a new simulator $S'$ for $P^{AR}$.
    An honest verifier samples the questions $x,y$ as $Q$ and other queries of the encodings. Then $S'$ runs $S$ to get answers $a,b$, computes a Mie's PCPP proof $\pi$ for $V(x,y,a,b) = 1$. As we have discussed above, $\pi$ can be calculated in time $\poly(n)$
    and is of $\poly(n)$ size. Hence the additional queries from the honest verifier about the Hadarmard encodings of $a,b$ and $\pi$ can also be answered in $\poly(n)$ time because $q'(n) = \poly(n)$ and answering the queries is simply calculating the inner product of the encodings with the queries.
    Therefore, $S'$ can sample from a new perfect correlation $p_{\ipt}'$ for $V^{AR}$.

    On the other hand, we let $S'$ answer the queries the same way even if $S$ outputs wrong answers $a,b$, so that 
    the output of $S'$ is a deterministic function of $S$'s outputs. 
    Then if $p'_{\ipt}$ is distinguishable from the perfect correlation, we can construct a new distinguisher for $p_{\ipt}$ by running the computation of $S'$ on $a,b$
    and feed the output to the distinguisher for $p'_{\ipt}$. This establishes the computational zero-knowledge property.
    For statistical zero-knowledge, if $p_{\ipt}$ is statistically close to the perfect correlation of $P$, then $p'_{\ipt}$ is also statistically close to the perfect correlation of $P^{AR}$ because even if $S$ generate wrong answers, there is a nonzero probability that $S'$ can answer $P^{AR}$'s queries correctly.
\end{proof}
\end{document}